\documentclass[lettersize,onecolumn]{IEEEtran}
\usepackage{amsmath,amsfonts}
\usepackage{amssymb, amsthm}
\usepackage{algorithmic}
\usepackage{algorithm}
\usepackage{array}
\usepackage{textcomp}
\usepackage{stfloats}
\usepackage{url}
\usepackage{verbatim}
\usepackage{graphicx}
\usepackage{cite}

\ifCLASSOPTIONcompsoc
\usepackage[caption=false,font=normalsize,labelfont=sf,textfont=sf]{subfig}
\else
\usepackage[caption=false,font=footnotesize]{subfig}
\fi

\makeatletter
\let\NAT@parse\undefined
\makeatother
\usepackage{hyperref}
\newtheorem{theorem}{Theorem}
\newtheorem{lemma}{Lemma}

\newtheorem{definition}{Definition}
\newtheorem{corollary}{Corollary}
\newtheorem{assumption}{Assumption}
\newtheorem{remark}{Remark}

\begin{document}

\title{Generalization Bounds for Transformer Channel Decoders}

\author{Qinshan Zhang, Bin Chen, Yong Jiang, and Shu-Tao Xia
\thanks{Qinshan Zhang, Yong Jiang, and Shu-Tao Xia are with Tsinghua Shenzhen International Graduate School, Shenzhen, China, and Pengcheng Laboratory, Shenzhen, China. Email: zhangqs24@mails.tsinghua.edu.cn, \{jiangy, xiast\}@sz.tsinghua.edu.cn.}
\thanks{Bin Chen is with Harbin Institute of Technology (Shenzhen), University Town of Shenzhen, Nanshan District, Shenzhen, 518055, China. Email: chenbin2021@hit.edu.cn. (Corresponding Author)}}



\maketitle

\begin{abstract}
Transformer channel decoders, such as the Error Correction Code Transformer (ECCT), have shown strong empirical performance in channel decoding, yet their generalization behavior remains theoretically unclear. This paper studies the generalization performance of ECCT from a learning-theoretic perspective. By establishing a connection between multiplicative noise estimation errors and bit-error-rate (BER), we derive an upper bound on the generalization gap via bit-wise Rademacher complexity. The resulting bound characterizes the dependence on code length, model parameters, and training set size, and applies to both single-layer and multi-layer ECCTs. We further show that parity-check–based masked attention induces sparsity that reduces the covering number, leading to a tighter generalization bound. To the best of our knowledge, this work provides the first theoretical generalization guarantees for this class of decoders.
\end{abstract}




\section{Introduction}

Deep neural networks (NNs) have been widely investigated in next-generation communication systems, with applications ranging from channel estimation to signal detection and resource allocation \cite{soltani2019DLbasedchannelestimation,hu2020DLforChannelEsitimationSurvey,arvinte2022mimoChannelEstimationusingGenerativeModels,khani2020adaptiveSignalDetectionforMIMO,gu2021knowledgeassistedDLin5Gto6G,jiang2022accurateChannelPredictionTransformer}. In channel coding, deep learning has also shown strong potential, where existing studies can be broadly categorized into two directions. One line of work explores data-driven end-to-end learning paradigms that jointly optimize the encoder and decoder, resulting in autoencoder–based architectures \cite{jiang2019turboAE,makkuva2021KOCodes,jamali2022productae,zhang2023adaptiveProductAE,choukroun2024learningECC}. The other focuses on improving the decoding performance of conventional error-correcting codes by adopting neural networks into decoding algorithms. Since the NP-hardness of the maximum-likelihood criterion limits the practical applicability of optimal decoding, designing efficient \emph{NN-based} soft-decision decoders capable of achieving high reliability remains a challenging and active research problem.

Existing NN-based decoders can be broadly categorized into model-based and model-free approaches. Model-based decoders reformulate conventional belief propagation (BP) decoding in an unrolled manner, yielding a parameterized trellis that can be interpreted as a feed-forward neural network. This network enables end-to-end learning of the decoding parameters, thereby adaptively balancing the relative importance of messages across iterations. Such decoders and their variants are commonly referred to as neural belief propagation (NBP)-like decoders \cite{nachmani2016learningNBPconference,nachmani2018deeplearningNBP,lugosch2017neuraloffsetNBP,buchberger2020pruningNBPJSAC,dai2021learningtodecodeprotographldpcNBP,chen2021cyclicallyNBP,zhang2024section-wiseNPBforQCLDPC}. They can be viewed as a generalization of the classical BP decoder, in the sense that setting all learnable weights to 1 recovers the standard BP algorithm. From a coding-theoretic perspective, BP decoders are widely regarded as achieving near-optimal asymptotic performance for appropriately designed codes; intuitively, this suggests similar behavior for NBP-like decoders.
From a learning-theoretic standpoint, recent studies have further established generalization bounds for NBP-like decoders, characterizing how their worst-case generalization behavior scales with the code parameters and the model complexity \cite{adiga2024generalizationboundforNBP}.

Model-free decoders extend the decoder architecture to general NNs. Decoding is performed by estimating an equivalent \emph{multiplicative} noise model in binary-input symmetric-output (BISO) channels \cite{richardson2002LDPCcapacityundermessagepassingdecoding}, which significantly mitigates overfitting to specific codewords in the training set \cite{bennatan2018syndromebasedapproach}. In particular, as a general and powerful neural architecture, the Transformer\cite{vaswani2017attentionisallyouneed} has demonstrated outstanding performance across a wide range of tasks, motivating its adoption for channel decoding. Indeed, empirical studies show that the Error Correction Code Transformer (ECCT) and its variants, as representative model-free decoders, achieve performance comparable to or even surpassing that of conventional decoding algorithms \cite{choukroun2022ECCT,park2025crossmpt,lau2025interplayBPandTransformerECCT,park2025multipleMasktforECCT}.

However, to the best of our knowledge, existing work lacks a theoretical characterization of the generalization performance of ECCT. Specifically, the generalization error is defined as the gap between the true bit-error-rate (BER) and the empirical BER achieved during training, which serves as a measure of the model’s tendency to overfit \cite{mohri2018foundationsofMachineLearningML}. This naturally raises a fundamental question: \emph{Given an ECCT, can we establish an upper bound on its generalization error that characterizes its worst-case guarantees on its performance on unseen codewords? Moreover, which parameters--such as the code parameters, model architecture, and training set size--influence this bound, and in what manner?}

In this paper, we first establish a theoretical connection between the error rate of multiplicative noise estimation and the BER of the decoded codeword. Building on this connection, we characterize an upper bound on the generalization error of ECCT via the Rademacher complexity associated with individual bit positions, referred to as the bit-wise Rademacher complexity. Following the established line of theoretical analyses for Transformer-based models \cite{edelman2022inductivebiasTransformerGeneralization,trauger2024sequence}, we begin with ECCTs consisting of a single attention layer, as this setting allows us to explicitly reveal how model and code parameters—other than the depth—affect the generalization bound. The results are then extended to multi-layer ECCT decoders. Our main contributions can be summarized as follows: (1) Based on learning-theoretic analysis, we characterize the bit-wise Rademacher complexity of ECCT with respect to the model weights and derive corresponding generalization bounds for both single-layer and multi-layer ECCTs, showing how these bounds scale with key parameters such as the code length, embedding dimension, and training set size. (2) Focusing on a key structural ingredient of ECCT, we theoretically reveal the benefit of sparsity induced by parity-check–based masked attention. In particular, it can significantly reduce the global Lipschitz bound with respect to the model weights, with a square-root dependence on the attention sparsity, which in turn decreases the covering number of the hypothesis space of the decoder function class and yields a tighter generalization bound compared to the unmasked version.

\section{Preliminaries}

\subsection{Problem Statement}
Let $\mathcal{X}$ be the label (or codeword) space,  $\mathcal{Y}$ be the sample space, and $\mathcal{G}$ be the function class defined as $\mathcal{G}=\{g: \mathcal{Y} \rightarrow \mathcal{X}\}$. Consider a decoder $g \in \mathcal{G}$ trained on a dataset $\{(\mathbf{y}_i,\mathbf{x}_i)\}_{i=1}^{m}$ of size $m$, using a BER loss function $l_{\text{BER}}$ defined as $l_{\text{BER}}\left(g(\mathbf{y}), \mathbf{x}\right)=\frac{1}{n}\sum_{j=1}^{n} \mathbb{I}(\hat{\mathbf{x}}[j] \neq \mathbf{x}[j])$, 
where $\mathbb{I}(\cdot)$ denotes the indicator function. 

The objective of a decoder is to minimize the empirical risk, defined as $ \hat{\mathcal{R}}_{\text{BER}}(g)=\frac{1}{m} \sum_{i=1}^{m} l_{\text{BER}}\left(g\left(\mathbf{y}_{i}\right), \mathbf{x}_{i}\right)$.
where $\mathbf{x}_i$ denotes the input and $y_{i}$ denotes the corresponding label. The loss function measures the difference between the prediction $g\left(\mathbf{y}_{i}\right)$ and the corresponding target $\mathbf{x}_{i}$. Let $\mathcal{D}$ be a distribution over $\mathcal{X}\times\mathcal{Y}$. 

The true risk is defined as $\mathcal{R}_{\text{BER}}(g)=\mathbb{E}_{(\mathbf{y}, \mathbf{x})\sim \mathcal{D}}\left[l_{\text{BER}}(g(\mathbf{y}), \mathbf{x})\right]$.

The generalization gap is defined as the difference between the empirical risk and the true risk: $\mathcal{R}_{\text{BER}}(g) - \hat{\mathcal{R}}_{\text{BER}}(g)$. The main goal of this paper is to derive an upper bound on this gap (i.e., generalization bound), and to characterize how this bound depends on the code and model parameters, and the training set size.

\subsection{Rademacher Complexity and Generalization Bound}

We briefly introduce the Rademacher complexity for a function class $\mathcal{F}$ and its relationship to the covering number and generalization bounds.

\begin{definition}[Empirical Rademacher Complexity]
    For a function class $\mathcal{F}$ trained with a dataset $\{(\mathbf{x}_i,\mathbf{y}_i)\}_{i=1}^{m}$ of size $m$, the empirical Rademacher complexity is defined as
    \begin{equation}
        R_{m}\left(\mathcal{F}\right) \triangleq \underset{\sigma}{\mathbb{E}}\left[\sup _{f \in \mathcal{F}} \frac{1}{m} \sum_{i=1}^{m} \sigma_{i} l\left(f\left(\mathbf{x}_{i}\right), \mathbf{y}_{i}\right)\right],
    \end{equation}
    where $\sigma_{i}$'s are i.i.d. Rademacher random variables, i.e. $\operatorname{Pr}\left(\sigma_{i}=1\right)=\operatorname{Pr}\left(\sigma_{i}=-1\right)=\frac{1}{2}$, and $l$ is the loss function.
\end{definition}

\begin{definition}[Covering Number]
      The covering number $\mathcal{N}\left(\mathcal{F}, \epsilon,\|\cdot\|_{k}\right)$ of the function class $\mathcal{F}$ with respect to the $k$-th norm for $\epsilon>0$ is defined as
     \begin{gather} 
        \mathcal{N}\left(\mathcal{F}, \epsilon,\|\cdot\|_{k}\right)=\min _{n}\left|\left\{g_{1}, \cdots, g_{n}\right\}\right|,  \\
        \text { s.t. } \min _{1 \leq i \leq n}\left\|f(\mathbf{x})-g_{i}(\mathbf{x})\right\|_{k} \leq \epsilon . \label{def: covering number}
     \end{gather}
     where \eqref{def: covering number} must be satisfied for any $f \in \mathcal{F}$ and input $\mathbf{x}$. The set $\left\{g_{1}, \cdots, g_{n}\right\} \subseteq \mathcal{F}$ is called the $\epsilon$-cover of $\mathcal{F}$.
\end{definition}

To relate the Rademacher complexity to generalization bounds via the covering number, a standard approach is to employ Dudley’s metric entropy integral. This integral admits several variants; here we present a modified version for bounded function classes.
\begin{lemma}[See \cite{bartlett2017spectrallynormalizedmarginbounds}]
For a real-valued function class $\mathcal{F}$ taking values in $[0,1]$, we have
    \begin{equation}
        R_{m}(\mathcal{F}) \leq \inf_{\alpha>0}\left(\frac{4 \alpha}{\sqrt{m}}+\frac{12}{m} \int_{\alpha}^{\sqrt{m}} \sqrt{\log \mathcal{N}\left(\mathcal{F}, \epsilon,\|\cdot \|_{2} \right)}d \epsilon\right).
        \label{eq: Rademacher by Dudley inter bound in PRELIMINARIES}
    \end{equation}
\end{lemma}

To bound the generalization gap for function class $\mathcal{F}$, a standard result can be obtained by probably approximately correct (PAC) learning theory \cite{mohri2018foundationsofMachineLearningML}. For any $\delta \in (0,1)$, with probability at least $1-\delta$, the generalization gap can be bounded as follows:
\begin{equation}
    \mathcal{R}_{\text{BER}}(f) - \hat{\mathcal{R}}_{\text{BER}}(f) \leq 2R_{m}(\tilde{\mathcal{F}})  +\sqrt{\frac{\log (1 / \delta)}{2 m}},
\end{equation}
where $\tilde{\mathcal{F}}$ is definde by the function class $\mathcal{F}$ and the loss function $l_{\text{BER}}$ as  $\tilde{\mathcal{F}} = \{(\mathbf{x}, \mathbf{y})\mapsto l(f(\mathbf{x}), \mathbf{y}) : g \in \mathcal{F} \}$.

\section{Main Results}
\subsection{The Pipeline of ECCT}

\begin{figure*}[!t]
	\centering
	\subfloat[The preprocessing for the received signal $\mathbf{y}$.]{\includegraphics[width=0.5\textwidth]{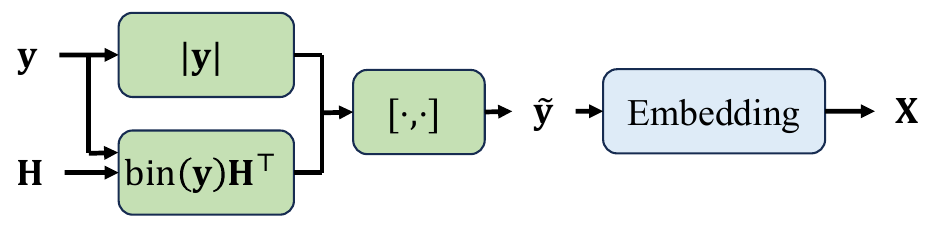}
		\label{fig: ecct preprocessing}}
	\hfill
    
	\subfloat[An attention layer]{\includegraphics[width=0.75\textwidth]{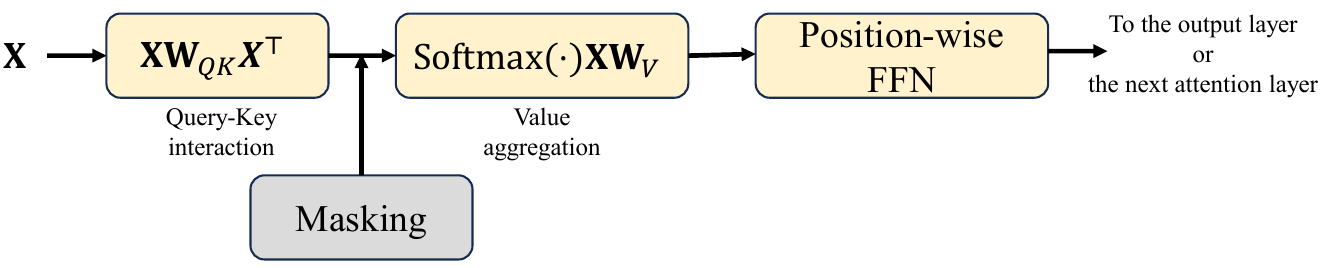}
		\label{fig: ecct attention layer}}
    \hfill

    \subfloat[The output layer]{\includegraphics[width=0.75\textwidth]{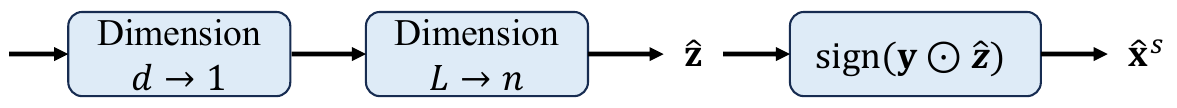}
		\label{fig: ecct output layer}}
	
	\caption{The pipeline of the ECCT decoder.}
	\label{fig: ecct pipeline}
\end{figure*}

To analyze ECCT, it is essential to formalize the flow of data through the model, upon which all subsequent analysis is built. Following established practices in the analysis of Transformer-based models \cite{edelman2022inductivebiasTransformerGeneralization,trauger2024sequence}, we begin with the single-layer, single-head attention setting, which can then be extended to more general configurations, such as multi-layer and multi-head attention. Note that in this context, a “layer” typically refers to a single attention layer in the Transformer architecture. The overall architecture of ECCT, depicted in Fig. \ref{fig: ecct pipeline}, comprises a preprocessing stage for the received signal $\mathbf{y}$, an attention layer, and an output layer. For the basic ECCT (Definition \ref{def: ECCT single layer}), we omit the “Masking” component shown in Fig. \ref{fig: ecct pipeline}, which will be introduced later in Definition \ref{def: masked attention of ECCT}.

During decoding, the input is first preprocessed by computing the magnitude of the received vector $\mathbf{y}$ and its syndrome, which are concatenated to form $\tilde{\mathbf{y}} \triangleq [|\mathbf{y}|, \operatorname{bin}(\mathbf{y})\mathbf{H}^{\top}]$, where $\operatorname{bin}(\cdot)$ denotes the hard-decision operation. Let the length of $\tilde{\mathbf{y}}$ be $L$, which in general equals the sum of the code length $n$ and the number of rows $r$ of $\mathbf{H}$. In typical settings, this reduces to $L=n+(n-k)=2n-k$. Each element is then embedded into a $d$-dimensional vector via $\mathbf{X} \triangleq \tilde{\mathbf{y}} \odot \mathbf{W}_{emb} $, where $\mathbf{W}_{emb}$ is a learnable real-valued matrix of size $L\times d$, and $\odot$ denotes elementwise multiplication (implicitly leveraging PyTorch’s broadcasting mechanism, which replicates $\tilde{\mathbf{y}}$ $d-1$ times to form an $L\times d$ matrix).

\begin{definition}[Basic (Single-Layer) ECCT]\label{def: ECCT single layer}
A single-layer ECCT decoder $f \in \mathcal{F}_{ECCT}$ consists of an attention layer followed by an output layer: 
\begin{enumerate}
    \item Attention layer: Following established practices in the analysis of Transformer \cite{edelman2022inductivebiasTransformerGeneralization}, we describe only the key operations, which include the self-attention mechanism and the feed-forward network (FFN) with a single hidden layer:
    \begin{equation}
        \mathbf{X}_{SA}=\left(\sigma\left[\left(\operatorname{Softmax}\left(\mathbf{X} \mathbf{W}_{Q} \mathbf{W}_{K}^{\top} \mathbf{X}^{\top}\right) \mathbf{X} \mathbf{W}_{V}\right) \mathbf{W}_{F1}\right]\right) \mathbf{W}_{F2}. \label{eq: self-attn, original, single layer}
    \end{equation}
    Let $\mathbf{W}_{QK}= \mathbf{W}_{Q} \mathbf{W}_{K}^{\top} \in \mathbb{R}^{d\times d}$, and let $\mathbf{W}_{V} \in \mathbb{R}^{d\times d}, \mathbf{W}_{F1} \in \mathbb{R}^{d\times ud}, \mathbf{W}_{F2} \in \mathbb{R}^{ud\times d}$, where $u$ denotes the scaling factor of the hidden layer in the FFN. The activation function $\sigma(\cdot)$ is assumed to be $L_{\sigma}$-Lipschitz, and the $\operatorname{Softmax}(\cdot)$ operator is $L_{sm}$-Lipschitz.

    \item Output layer: The role of the output layer is to align the data dimension with the required output size. Specifically, it performs (a) a reduction of each embedded vector from $d$ dimensions to 1, and (b) a shortening of the sequence length from $L$ to $n$. Formally, this can be expressed as: 
    \begin{equation}
        \begin{aligned}
            \hat{\mathbf{z}} &= (\mathbf{X}_{SA}\mathbf{W}_{o1})^{\top}\mathbf{W}_{o2} \\
            &=\mathbf{W}_{o1}^{\top}\mathbf{X}_{SA}^{\top}\mathbf{W}_{o2},
        \end{aligned}
    \end{equation}
    where $\mathbf{W}_{o1}\in \mathbb{R}^{d\times 1}, \mathbf{W}_{o2}\in \mathbb{R}^{L\times n}$. For the output $\hat{z}$, its $j$-th element can be written as: 
    \begin{equation}
        \hat{\mathbf{z}}[j]=\mathbf{W}_{o1}^{\top}\mathbf{X}_{SA}^{\top}\mathbf{W}_{o2}[:,j],
    \end{equation}
    where $\mathbf{W}_{o2}[:,j]$ denotes the $j$-th column of $\mathbf{W}_{o2}$. During the decision stage, the entries of $\hat{\mathbf{z}}$ are normalized to the interval $[0,1]$ through the sigmoid function.
\end{enumerate}
\end{definition}

It is important to note that the output of ECCT, $\hat{\mathbf{z}}=f(\mathbf{X})$, estimates the \emph{multiplicative} noise applied to $\mathbf{x}^{s}$, the codeword $\mathbf{x}$ modulated by binary phase shift keying (BPSK), rather than directly estimating the \emph{additive} noise in the AWGN channel\cite{bennatan2018syndromebasedapproach,choukroun2022ECCT}. These two formulations are theoretically equivalent in Binary-Input Symmetric-Output (BISO) channels\cite{richardson2002LDPCcapacityundermessagepassingdecoding}. It can be proved that a decoder that takes $\tilde{\mathbf{y}}$ as input and estimates the multiplicative noise can, in the ideal case, achieve maximum a posteriori (MAP) decoding. Moreover, this formulation decouples the decoder from the specific transmitted codeword, mitigating the risk of severe overfitting during training and enabling the use of only noisy all-zero codewords as the training dataset. Further details can be found in \cite{bennatan2018syndromebasedapproach}.

\subsection{Generalization Bound via Rademacher Complexity}
The overall decoder class $\mathcal{G}$ produces the recovered codeword, given by $g(\mathbf{y})=\hat{\mathbf{x}}^{s}=\operatorname{sign}(\mathbf{y} \odot f(\mathbf{X}))$ for a $g \in \mathcal{G}$, where $\operatorname{sign}(\cdot)$ applies element-wise binarization.The recovered codeword $\hat{\mathbf{x}}^{s}\in \{+1, -1\}^{n}$ and the superscript ``s" indicates that each entry takes values in $\{+1,-1\}$, which are in one-to-one correspondence with the underlying bit sequence. Since ECCT does not directly estimate the codeword itself, we establish the connection between noise estimation and codeword recovery through the following lemma. Hereafter, the term ``decoder” is used interchangeably to denote the ECCT and the overall decoder.

\begin{lemma}\label{lemma: equivalence of noise vs codeword BER}
    Under the loss function defined by the \emph{codeword} bit error rate (BER) of the entire decoder $g \in \mathcal{G}$, the bit error probability of the \emph{noise} estimated by the ECCT is equivalent to the decoder's BER.
\end{lemma}

\begin{proof}[Proof of Lemma~\ref{lemma: equivalence of noise vs codeword BER}]
For a given codeword $\mathbf{x}$, let the decoder’s estimate be $\hat{\mathbf{x}}$, and let their respective binarized forms be $\mathbf{x}^{s}$ and $\hat{\mathbf{x}}^{s}$. The decoder output satisfies $g(\mathbf{y})=\mathbf{x}^{s}=\operatorname{sign}(\mathbf{y}\odot f(\mathbf{X}))=\operatorname{sign}(\mathbf{x}^{s}\odot\mathbf{z}\odot\hat{\mathbf{z}})=\mathbf{x}^{s}\odot\mathbf{z}^{s}\odot\hat{\mathbf{z}}^{s}$. 
The BER loss associated with this codeword is given by: 
    \begin{equation}
    \begin{aligned}
        l_{\text{BER}}(g(\mathbf{y}),\mathbf{x}^{s}) &= \frac{\sum_{j=1}^{n}\mathbb{I}(\hat{\mathbf{x}}[j]\neq \mathbf{x}[j])}{n} \\
        &= \frac{\sum_{j=1}^{n}\mathbb{I}(\hat{\mathbf{x}}^{s}[j]\neq \mathbf{x}^{s}[j])}{n} \\
        &= \frac{\sum_{j=1}^{n}\mathbb{I}((\mathbf{x}^{s}[j]\cdot\mathbf{z}^{s}[j]\cdot\hat{\mathbf{z}}^{s}[j])\neq \mathbf{x}^{s}[j])}{n} \\
        &= \frac{\sum_{j=1}^{n}\mathbb{I}(\hat{\mathbf{z}}^{s}[j]\neq \mathbf{z}^{s}[j])}{n} \\
        &= \frac{\sum_{j=1}^{n}\mathbb{I}(\hat{\mathbf{z}}[j]\neq \mathbf{z}[j])}{n} \\
        &= l_{\text{BER}}(f(\mathbf{X}), \mathbf{z})
    \end{aligned}
    \end{equation}
\end{proof}

Lemma \ref{lemma: equivalence of noise vs codeword BER} shows that the binary classification task on each bit of the codeword is equivalent to the binary classification task on the estimated noise, where the labels correspond to the hard-decision multiplicative noise $\mathbf{z}$. From the perspective of hard-decision decoding, ECCT essentially estimates the error pattern of the codeword \cite{ryan2009channel}, whose per-bit error probabilities coincide with the BER of the codeword. This provides an alternative interpretation of the lemma. Based on this lemma, we naturally obtain the empirical Rademacher complexity of the $j$-th output of ECCT:
\begin{equation}
        R_{m}\left(\mathcal{F}_{\text{ECCT}}[j]\right) \triangleq \underset{\sigma}{\mathbb{E}}\left[\sup _{f \in \mathcal{F}_{\text{ECCT}}} \frac{1}{m} \sum_{i=1}^{m} \sigma_{i} \cdot f\left(\mathbf{y}_{i}\right)[j]\right].
        \label{def: noise bit-wise Rad complexity}
\end{equation}

As a direct consequence of Lemma \ref{lemma: equivalence of noise vs codeword BER} and the result in \cite[Prop.~1]{adiga2024generalizationboundforNBP}, we obtain the following result for ECCT:

\begin{corollary}\label{prop: ECCT generalization bound vs bit-wise Rad complexity}
    For any $ \delta \in (0, 1)$, with probability at least $1-\delta$, the generalization gap for any ECCT decoder $f \in \mathcal{F}_{\text{ECCT}}$ can be upper bounded as follows:
    \begin{equation}
        \mathcal{R}_{\mathrm{BER}}(f)-\hat{\mathcal{R}}_{\mathrm{BER}}(f) \leq \frac{1}{n} \sum_{j=1}^{n} R_{m}\left(\mathcal{F}_{\text{ECCT}}[j]\right)+\sqrt{\frac{\log (1 / \delta)}{2 m}},
        \label{eq: generalization bound via Rademacher(ECCT)}
    \end{equation}
    where $R_{m}\left(\mathcal{F}_{\text{ECCT}}[j]\right)$ denotes the bit-wise Rademacher complexity for the $j$-th output bit, and $\mathcal{F}_{\text{ECCT}}$ is the funtction class of ECCT decoders.
\end{corollary}


We now present the main result, which provides a generalization bound expressed in terms of the key parameters of the model and code based on Rademacher complexity.

\begin{assumption}\label{assumption: bounded input and weight}
    We make the following assumptions:
    \begin{itemize}
        \item Input: For each row of $\mathbf{X} \in \mathbb{R}^{L \times d}$ (i.e. $\mathbf{x}_{l}, l=1,\ldots,L$, denoting each positional embedding vector), it is bounded by $\left\|\mathbf{x}_{l}\right\|_{2}\leq b_{x}$.
        \item Weight: The weight matrices, $\mathbf{W}_{QK}, \mathbf{W}_{V}, \mathbf{W}_{F1}, \mathbf{W}_{F2}, \mathbf{W}_{emb}$, are bounded by their corresponding spectral norm bounds, $B_{QK}, B_{V}, B_{F1}, B_{F2}, B_{emb}$, respectively. The vectors $\mathbf{W}_{o1}$ and $\mathbf{W}_{o2}[:,j]$ are bounded by $L_2$ norm bounds $b_{o1}$ and $ b_{o2}$, respectively. Each entry of the weight matrices is assumed to be bounded in magnitude by $w$.
    \end{itemize}
\end{assumption}

\begin{theorem}\label{theorem: bound for single-layer ecct (main result)}
    For any ECCT $f \in \mathcal{F}_{ECCT}$ and any $\delta \in (0,1)$, with probability at least $1-\delta$, the generalization gap can be bounded as follows:
    \begin{equation}
        \mathcal{R}_{\mathrm{BER}}(f)-\hat{\mathcal{R}}_{\mathrm{BER}}(f) \leq \frac{4}{\sqrt{m}}+\sqrt{\frac{\log (1 / \delta)}{2 m}}+12 \sqrt{\frac{(L+(2 u+2) d^{2}) \log (18 \sqrt{m d}B L^{2})}{m}},
    \end{equation}
    where $L=n+r$ denotes the length of the input sequence, $d$ denotes the embedding dimension, $u$ denotes the scaling factor of the hidden layer in the FFN, $m$ denotes the training dataset size, and $B=b_{o1}^{2} L_{\sigma} B_{F1}^{2} B_{F2}^{2} B_{V}^{2}B_{QK} b_{x}^{3}w^{2}$.
\end{theorem}

\begin{proof}[Proof-Sketch of Theorem~\ref{theorem: bound for single-layer ecct (main result)}]
    The main idea of the proof is that the Rademacher complexity can be upper-bounded via Dudley’s entropy integral, which in turn depends on the covering number. The covering number of ECCT characterizes the minimal cardinality of a subset of $\mathcal{F}_{ECCT}$ required to approximate the decoding function to a prescribed accuracy. Its covering construction decomposes into a Cartesian product of coverings for the individual weight matrices of the decoder \cite{chen2020generalizationofafamilyofRNN,edelman2022inductivebiasTransformerGeneralization}. To derive upper bounds on the covering numbers of these matrices, we first show that ECCT is Lipschitz continuous with respect to each weight matrix. In particular, a small perturbation to any given weight matrix induces only a controlled change in the output of the function class, with the change governed by the corresponding Lipschitz constant. We then use the differentiability of multivariate continuous functions \cite[Theorem 3.1.6]{federer2014geometricmeasure} and the fact that Lipschitz constants can be bounded via gradients to obtain explicit upper bounds on these constants \cite{virmaux2018lipschitzConstantEstimation}. Finally, these bounds allow us to sequentially establish the covering numbers of the individual matrices, the covering number of ECCT, the bit-wise Rademacher complexity, and ultimately the generalization bound.
    
        
\end{proof}

Next, we show that for ECCT with a parity-check–matrix–based masking operation, the result of Theorem \ref{theorem: bound for single-layer ecct (main result)} continues to hold asymptotically. Specifically, we first define the masked attention mechanism and then present the generalization bound for ECCT under the masking operation.

\begin{algorithm}[t]
    \caption{Mask Matrix Construction, i.e. $\operatorname{mask}(\mathbf{H})$}
    \label{algo: mask matrix construction} 
    \begin{algorithmic}[1]
        \REQUIRE Parity-check matrix $\mathbf{H}$ of size $r \times n$
        \ENSURE Mask matrix $\mathbf{M}$ of size $L \times L$
        \STATE $\mathbf{M} \gets \mathbf{I}_{L}$ \COMMENT{Initialization}
        \FOR {$i \gets 1 \  \TO \  r$}
            \FOR {$j \gets 1 \  \TO \  n$}
                \IF{$\mathbf{H}[i,j]=1$}
                    \STATE $\mathbf{M}[i+n,j], \mathbf{M}[j,i+n] \gets 1, 1$
                \ENDIF
                \FOR {$k \gets 1 \  \TO \  n$}
                    \IF{$\mathbf{H}[i,j]=1$ and $\mathbf{H}[i,k]=1$}
                        \STATE $\mathbf{M}[j,k], \mathbf{M}[k,j] \gets 1, 1$
                    \ENDIF
                \ENDFOR
            \ENDFOR
        \ENDFOR
        \STATE $\mathbf{M}\gets \neg \mathbf{M}$ \COMMENT{Showing the entries to be masked}
        \FOR {$i \gets 1 \  \TO \  L$}
            \FOR {$j \gets 1 \  \TO \  L$}
                \IF{$\mathbf{M}[i,j]=1$}
                    \STATE $\mathbf{M}[i,j]\gets -\infty$
                \ENDIF
            \ENDFOR
        \ENDFOR
    \end{algorithmic} 
\end{algorithm}

\begin{figure}[!t]
    \centering
    \includegraphics[width=0.7\textwidth]{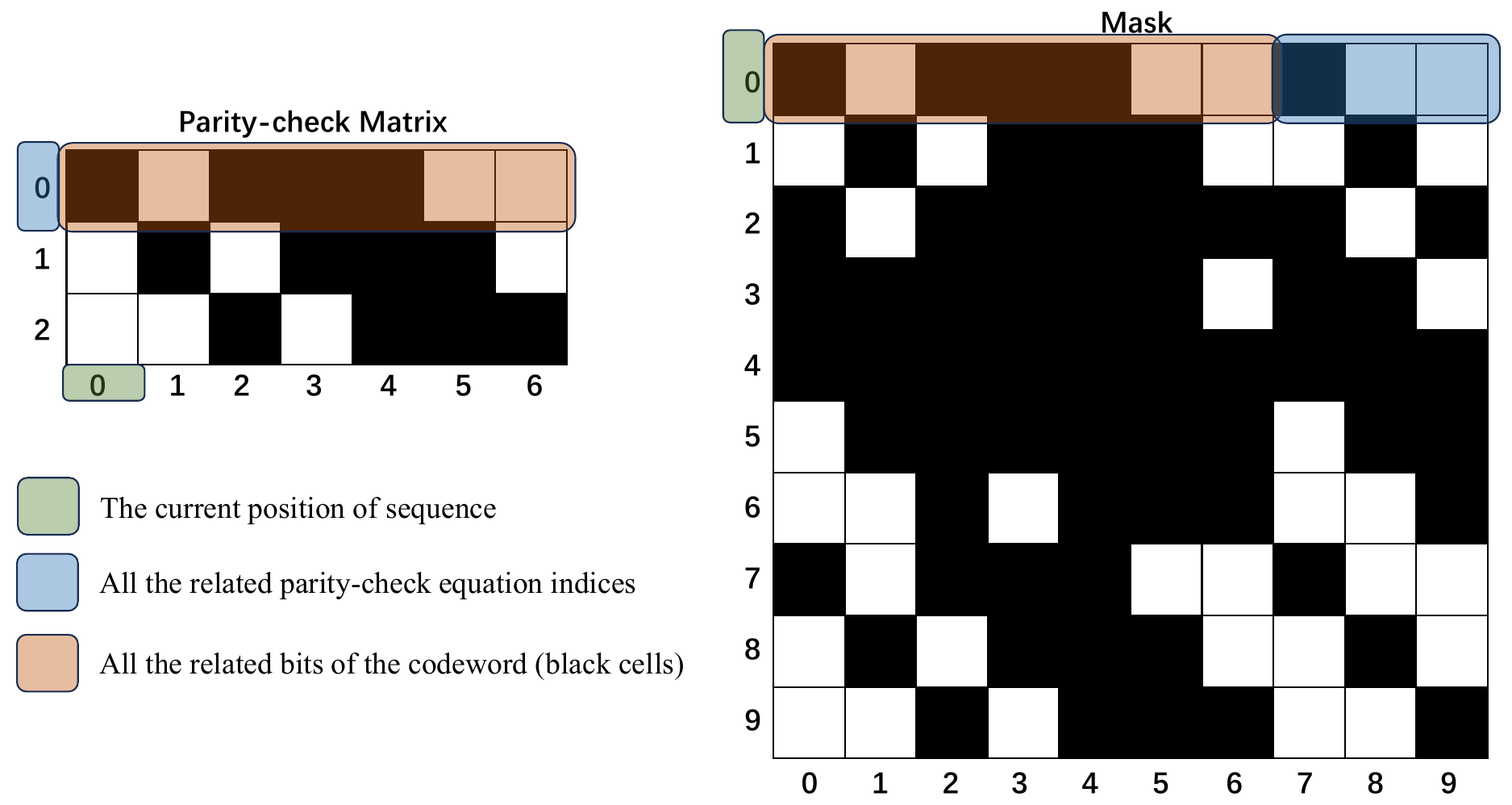}
    \caption{The mask of the $(7,4)$ Hamming code. For the length $L=7+3=10$ sequence, the first 7 positions represent codeword bits and the last 3 represent parity-check equations. The green box marks index 0 of the sequence, which we focus on; the blue shows its check-node connections, and the orange shows the associated bit-to-bit interactions induced by the parity-check constraints.}
    \label{fig: mask example Hamming (7,4)}
\end{figure}

\begin{definition}[Sparse Attention of ECCT]\label{def: masked attention of ECCT}
    In an attention layer, the masked attention operation can be defined as:
    \begin{equation}
        \mathbf{A}^{mask}\triangleq\operatorname{Softmax}\left(\mathbf{X} \mathbf{W}_{QK} \mathbf{X}^{\top}+\operatorname{mask}(\mathbf{H})\right),
        \label{eq: masked attention}
    \end{equation}
    where $\operatorname{mask}(\mathbf{H})$ denotes the $L \times L$ mask constructed from the parity-check matrix $\mathbf{H}$. The construction of this mask is given by Algorithm \ref{algo: mask matrix construction}. We further define the resulting matrix $\mathbf{A}^{mask}$ as $P$-sparse attention. For any row index $i$, it satisfies:
    \begin{equation}
        |\{j|\mathbf{A}^{mask}[i,j]\neq 0\}|\leq P.
    \end{equation}
    Note that $\mathbf{A}^{mask}$ is also a symmetric matrix.
\end{definition}

After query-key interaction $\mathbf{X} \mathbf{W}_{QK} \mathbf{X}^{T}$, the resulting $L \times L$ matrix represents the attention scores between all positions in the sequence. A higher score corresponds to a higher level of ``attention", and thus a more influential role in subsequent computations. Moreover, after applying the Softmax operation to the $i$-th row, the resulting distribution can be interpreted as the proportion of attention that position $i$ allocates to every position in the sequence. Fig. \ref{fig: mask example Hamming (7,4)} shows, for the $(7,4)$ Hamming code, the correspondence between the parity-check matrix and the attention mask, where black (white) cells denote preserved (masked) positions.

Overall, the masked attention mechanism forces the attention scores between positions that are not structurally coupled by any parity-check constraint to become zero after the row-wise Softmax operation. This compels the model to focus on positions that do satisfy parity-check constraints, which empirically leads to substantial performance improvements compared with ECCT without masking \cite{choukroun2022ECCT}. However, we emphasize that ECCT with sparse attention and the basic ECCT exhibit the same asymptotic order in their generalization bounds, as established in Theorem 2.

\begin{theorem}[Impact of Sparse Attention]\label{theorem: impact of mask}
    For any ECCT $f \in \mathcal{F}_{ECCT}$ equipped with $P$-sparse attention and any $\delta \in (0,1)$, with probability at least $1-\delta$, the generalization gap can be bounded as follows:
    \begin{equation}
        \mathcal{R}_{\mathrm{BER}}(f)-\hat{\mathcal{R}}_{\mathrm{BER}}(f) \leq \frac{4}{\sqrt{m}}+\sqrt{\frac{\log (1 / \delta)}{2 m}}+12 \sqrt{\frac{(L+(2 u+2) d^{2}) \log (18 \sqrt{m d}\Lambda^{sparse})}{m}},
    \end{equation}
    where $\Lambda^{sparse}$ represents the global Lipschitz bound of the sparse version of the ECCT decoder. Specifically, $\Lambda^{sparse}$ is dominated by the Lipschitz constant with respect to the query-key interaction term $L_{QK}$, which is defined as $\Lambda^{sparse}=BL^{1.5}\sqrt{P}$, where $B$ is the constant defined in Theorem \ref{theorem: bound for single-layer ecct (main result)}. Compared to the global Lipschitz bound of the basic (or dense) version $\Lambda^{dense}=BL^{2}$ in Theorem \ref{theorem: bound for single-layer ecct (main result)}, the sparsity constraint induces a complexity contraction factor:
    \begin{equation}
        \eta = \frac{\Lambda^{sparse}}{\Lambda^{dense}}\leq \sqrt{\frac{P}{L}}.
    \end{equation}
\end{theorem}


\begin{proof}[Proof-Sketch of Theorem~\ref{theorem: impact of mask}]
    The proof largely follows the framework of Theorem \ref{theorem: bound for single-layer ecct (main result)}. The critical deviation lies in the analysis of the backward propagation through the attention mechanism. We explicitly prove that the gradient of the loss with respect to the attention scores strictly inherits the sparsity structure defined by the mask $\operatorname{mask}(\mathbf{H})$. Leveraging this property and the differentiability of multivariate continuous functions, we tighten the spectral norm upper bound in the Lipschitz constant estimation for the query-key interaction term, thereby reducing the effective radius of the covering number. Detailed derivations and the proof of gradient sparsity are provided in Appendix \ref{appendix: proof of impact of mask (theorem2)}.
\end{proof}


\subsection{Generalization Bound for Multi-Layer ECCT}

Having established the generalization bound for the single-layer ECCT, we now turn to the multi-layer case. A multi-layer ECCT is a direct extension of the single-layer model, obtained by stacking multiple attention layers.

\begin{definition}[Multi-Layer ECCT]\label{def: multi-layer ECCT}
    Keeping with the definitions and notation previously set forth, we will define ECCT $f^{(T)}$ in $\mathcal{F}_{ECCT,T}$ with $T$ attention layers. Let $\mathcal{W}^{(i)}=\{\mathbf{W}_{QK}^{(i)},\mathbf{W}_{V}^{(i)},\mathbf{W}_{F1}^{(i)},\mathbf{W}_{F2}^{(i)}\}$, and we define the $i$-th attention layer as:
    \begin{equation}
        \begin{aligned}
        \mathbf{X}^{(i)}&=\phi^{(i)}\left(\mathbf{X}^{(i-1)};\mathcal{W}^{(i)}\right)\\  &=\left(\sigma\left[\left(\operatorname{Softmax}\left(\mathbf{X}^{(i-1)} \mathbf{W}_{QK}^{(i)} (\mathbf{X}^{(i-1)})^{\top}\right) \mathbf{X}^{(i-1)} \mathbf{W}_{V}^{(i)}\right) \mathbf{W}_{F1}^{(i)}\right]\right) \mathbf{W}_{F2}^{(i)},
        \end{aligned}
    \end{equation}
    where $\mathbf{X}^{(i-1)}(i=1,\ldots,T)$ denotes the input of the attention layer, and $\mathbf{X}^{(0)}$ denotes the input of the ECCT decoder. For the $j$-th bit output, it can be written as $\hat{\mathbf{z}}[j]=\mathbf{W}_{o 1}^{\top} (\mathbf{X}^{(T)})^{\top} \mathbf{W}_{o 2}[:, j]$, by keeping the same with Definition \ref{def: ECCT single layer}.
\end{definition}

\begin{theorem}[Generalizaiton Bound for Multi-Layer ECCT]\label{theorem: multi-layer ecct bound}
    For any $T$-layer ECCT $f^{(T)} \in \mathcal{F}_{ECCT,T}$ equipped with $P$-sparse attention and any $\delta \in (0,1)$, with probability at least $1-\delta$, the generalization gap can be bounded as follows:
    \begin{equation}
        \mathcal{R}_{\mathrm{BER}}(f)-\hat{\mathcal{R}}_{\mathrm{BER}}(f) \leq \frac{4}{\sqrt{m}}+\sqrt{\frac{\log (1 / \delta)}{2 m}}+12 \sqrt{\frac{(L+(2 u+2) d^{2} T) \log (6 \sqrt{m d}\Lambda^{(T)}) }{m}},
    \end{equation}
    where $\Lambda^{(T)}=b_{o1}L_{\sigma}B_{V}B_{F1}B_{F2}b_{x}^{3}wL^{1.5}\sqrt{P}\left(\sqrt{P}B_{V}B_{F1}B_{F2}(1+2B_{QK}b_{x}^{2})\right)^{T-1}$ represents the global Lipschitz bound for the ECCT decoder. Furthermore, compared to the global Lipschitz bound for the dense version of the $T$-layer ECCT decoder, the sparsity constraint induces a complexity contraction factor:
    \begin{equation}
        \eta^{(T)} \leq \left(\sqrt{\frac{P}{L}}\right)^{T}.
    \end{equation}
\end{theorem}

\begin{proof}[Proof-Sketch of Theorem~\ref{theorem: multi-layer ecct bound}]
    The proof follows the analytical trajectory established above. However, the extension to the multi-layer setting ($T>1$) introduces recursive dependencies using the chain rule. Since the decoder $f^{(T)}$ is a composition of $T$ attention layers, the sensitivity of the output with respect to the weights at the $i$ layer in $\mathcal{W}^{(i)}$, depends on the gradient backpropagation path through all subsequent layers $t \in \{t+1, \ldots, T\}$. Moreover, we leverage the structural alignment between the gradient sparsity pattern and the attention mask $\operatorname{mask}(\mathbf{H})$. Specifically, compared to the dense (or unmasked) attention, the $P$-sparse attention results in a global complexity reduction scaling with $\left(\sqrt{P/L}\right)^{T}$ for the bottom-most layer ($i=1$). Substituting this tightened global Lipschitz bound into the generalization bound framework completes the proof.
\end{proof}


In the presence of additive white Gaussian noise (AWGN), the input $\mathbf{X}$ is theoretically unbounded. We investigate the effect of the noise level under the AWGN channel on the generalization bound, as stated in Theorem \ref{theorem: bound for AWGN}.

\begin{theorem}[Generalization Bound for AWGN Channel]\label{theorem: bound for AWGN}
     For any $T$-layer ECCT $f^{(T)} \in \mathcal{F}_{ECCT,T}$ equipped with $P$-sparse attention, when the input is unbounded, and any $\delta \in (0,1)$, with probability at least $1-\delta$, the generalization gap can be bounded as follows:
    \begin{equation}
        \begin{aligned}
        \mathcal{R}_{\mathrm{BER}}(f)-\hat{\mathcal{R}}_{\mathrm{BER}}(f) \leq &\frac{4}{\sqrt{m}}+\sqrt{\frac{\log (1 / \delta)}{2 m}}+ \\
        &\min_{b_x} \left\{\mathcal{B}(T, L, P, d, m, b_x) + Pr(\exists i, \left\|\mathbf{X}[i,:]\right\|_2>b_x)\right\},
        \end{aligned}
    \end{equation}
    where $\mathcal{B}(T, L, P, d, m) =  12 \sqrt{\frac{(L+(2 u+2) d^{2} T) \log (6 \sqrt{m d}\Lambda^{(T)}) }{m}}$. Under an AWGN channel with noise variance $\rho^2$, using BPSK modulation, we have $Pr(\exists i, \left\|\mathbf{X}[i,:]\right\|_2>b_x) = n \left[Q\left(\frac{b_{x}-B_{e m b}}{\rho B_{e m b}}\right)+Q\left(\frac{b_{x}+B_{e m b}}{\rho B_{e m b}}\right)\right]$.
\end{theorem}

\begin{proof}[Proof-Sketch of Theorem~\ref{theorem: bound for AWGN}]
    Unbounded input violates the standard Lipschitz continuity requirement for global generalization analysis. To resolve this, we partition the sample space into the regime $\mathcal{E}=\left\{\forall i \in[L],\left\|\mathbf{X}[i,:]\right\|_{2} \leq b_{x}\right\}$ and its complement $\mathcal{E}^c$. While the generalization gap in $\mathcal{E}$ is controlled by the Rademacher complexity of the ECCT, the gap in $\mathcal{E}^c$ cannot be bounded by Lipschitz-based methods. For the regime $\mathcal{E}^c$, we analyze the BPSK-modulated Gaussian inputs, where the probability for any $i$-th position of the input sequence is governed by the Q-function ($i=1,\ldots,n$). Then, we apply the union bound over the first $n$ positions. The final theorem is obtained by minimizing the sum of the bounded complexity term and the probability for $\mathcal{E}^c$ with respect to $b_x$.
\end{proof}

\begin{remark}
    Here, we distinguish the roles of the input bound $b_x$ and the spectral norm bound $B_{emb}$ of the embedding weight matrix in Theorem \ref{theorem: bound for AWGN}. The term $b_x$ represents a variational optimization in the theoretical analysis: since the inequality holds for any arbitrary $b_x>0$, the tightest guarantee is given by the specific $b_{x}^{*}$ that balances the bounded complexity term $\mathcal{B}(T, L, P, d, m)$ and the probability term $Pr(\exists i, \left\|\mathbf{X}[i,:]\right\|_2>b_x)$. In contrast, $B_{emb}$ is a structural constraint determined by the network initialization and regularization. For a fixed trained model, $B_{emb}$ is a constant. Consequently, the bound implies that the generalization gap is controlled not only by the tail behavior of the AWGN noise but also by the magnitude of the embedding weights. This provides a theoretical justification for enforcing $\|\mathbf{W}_{emb}\|_2\leq B_{emb}$ via weight decay \cite{krogh1991simpleweightdecay,loshchilov2017decoupledweightdecay} or spectral normalization \cite{miyato2018spectralnormalization} during training in practice to control the complexity term effectively.
\end{remark}

\subsection{Asymptotic Scaling and Insights}


\textbf{Code parameters}: As implied by Theorem \ref{theorem: bound for single-layer ecct (main result)}, the generalization gap scales with the input sequence length $L$ as $\mathcal{O}(\sqrt{L})$. Furthermore, $L$ is determined jointly by the code blocklength $n$ and the number of rows $r$ of the parity-check matrix. It is typically assumed that the number of parity-check equations does not exceed the blocklength, which implies $n<L<2n$, and hence $L$ scales linearly with $n$ in the asymptotic regime. In the standard ECCT setting \cite{choukroun2022ECCT}, the parity-check matrix has $r=n-k$ rows, leading to $L=n(2-k/n)=n(2-R)$, where $R$ denotes the code rate. Consequently, as $n$ becomes large, the generalization error bound exhibits an asymptotic growth rate of $\mathcal{O}(\sqrt{n})$. Moreover, the generalization gap scales with the training set size $m$ as $\mathcal{O}(\frac{1}{\sqrt{m}})$. This implies that maintaining a fixed generalization accuracy for longer codes requires the training set size to scale at least linearly with the blocklength.

\textbf{Number of Attention Layers}: Extending the single-layer result to the multi-layer setting, Theorem \ref{theorem: multi-layer ecct bound} shows that increasing depth amplifies the Lipschitz constant multiplicatively across layers, leading to an accumulated effect in the covering number. Consequently, the generalization bound scales with $T$ as $\mathcal{O}(T)$, reflecting the classical depth–complexity trade-off in deep architectures: deeper ECCTs can model more expressive decoding functions, but this increased expressiveness comes at the cost of a larger hypothesis space and thus weaker generalization guarantees unless compensated by more data or stronger regularization, i.e., by increasing $m$ or enforcing tighter bounds on the norms of the weight matrices that control the network’s Lipschitz constant during training \cite{krogh1991simpleweightdecay,loshchilov2017decoupledweightdecay,miyato2018spectralnormalization}.

\textbf{Sparse (Masked) Attention}: The sparsity level $P$ explicitly reduces the dependency of the complexity term on the sequence length $L$. Specifically, since each query interacts with at most $P$ keys in $\mathbf{A}^{mask}$, the column sum of the attention matrix is bounded by $P$, tightening its spectral norm bound to $\sqrt{P}$. Consequently, the complexity factor inside the logarithmic covering number, namely the global Lipschitz bound $\Lambda^{(T)}$, is reduced from $\mathcal{O}(L^{2T})$ to $\mathcal{O}(L^{1.5}P^{T/2})$. This theoretical result suggests that in long-code tasks where $L\gg P$, the sparse model effectively reduces the volume of the hypothesis space by limiting the interaction range of each position of the sequence, thereby mitigating the risk of overfitting. Several ECCT variants \cite{park2025crossmpt,park2025multipleMasktforECCT} exploit mathematically equivalent forms of parity-check equations to induce sparser attention masks and have shown empirical gains, inspired by the success of sparse attention in Transformers. However, a theoretical explanation for this decoder design has been missing. Our results provide the first learning-theoretic justification: parity-check–induced sparsity provably tightens the generalization bound, with a gain that grows exponentially with the depth $T$.

\section{Experiments}

\begin{figure}[!t]
    \centering
    \includegraphics[width=0.7\textwidth]{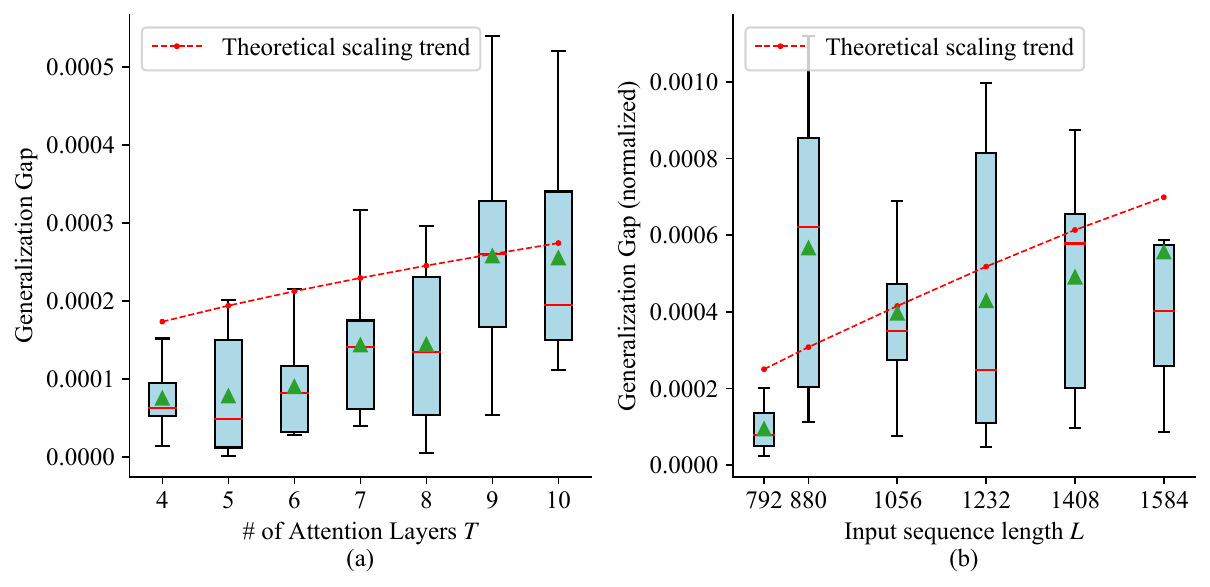}
    \caption{Generalization gap versus (a) number of attention layers $T$ and (b) input sequence length $L$. The theoretical scaling trend is obtained by substituting the corresponding numerical values.}
    \label{fig: bound visualization}
\end{figure}

In Fig. \ref{fig: bound visualization}, we report numerical results to corroborate the theoretical scaling laws. We consider BPSK transmission over an AWGN channel with $E_b/N_0=2$ dB, embedding dimension $d=32$, and training set size $m=12{,}800$. Each experiment is repeated over 10 independent random trials, and the results are summarized using boxplots. The generalization gap (difference between average training BER and testing BER) is evaluated as a function of (a) the number of attention layers $T$ and (b) the input sequence length $L$. In Fig. \ref{fig: bound visualization} (a), we use the CCSDS $(128,64)$ code to study the scaling with $T$. In Fig. \ref{fig: bound visualization} (b), with $T=1$, we examine the dependence on $L$ using a family of descendant codes derived from the WiMAX $(1056,528)$ LDPC code by masking the rightmost double-diagonal structure of the parity-check matrix, thereby largely preserving the structural properties of the parent code. To remove the effect of code rate variation, the generalization gap is normalized by the empirical BER \cite{adiga2024generalizationboundforNBP}. The observed trends are consistent with the theoretical predictions of Theorems \ref{theorem: bound for single-layer ecct (main result)} and \ref{theorem: multi-layer ecct bound}. In particular, both the empirical results and the theoretical bounds exhibit comparable monotonic behaviors with respect to $T$ and $L$. Due to the finite experimental regime, the results are not intended to identify the exact asymptotic exponents, but rather to validate the qualitative scaling behavior predicted by the theory.

\section{Conclusion}
In this paper, we provided the first learning-theoretic generalization guarantees for Transformer-based neural decoders. By establishing a connection between multiplicative noise estimation errors and the decoded BER, we derived upper bounds on the generalization gap of ECCT via the bit-wise Rademacher complexity. The analysis characterizes how the bound scales with key code and model parameters, and extends naturally from single-layer to multi-layer architectures. Moreover, we showed that parity-check–based masked attention induces sparsity that provably tightens the global Lipschitz bound, reduces the effective hypothesis space, and yields a strictly tighter generalization bound compared to unmasked attention. These results offer theoretical insights into the scalability and design of ECCT decoders from a generalization perspective.

\appendices

\section{Proof of Theorem \ref{theorem: bound for single-layer ecct (main result)}}\label{appendix: proof of theorem 1}

Before proving Theorem \ref{theorem: bound for single-layer ecct (main result)}, we first present Lipschitz constants with respect to each weight matrix of an ECCT defined in Definition \ref{def: ECCT single layer} (details can be found in Appendix \ref{appendix: bounding Lips constants}). These results are subsequently used to obtain the covering numbers in the weight space, and finally to derive the Rademacher complexity and the corresponding generalization bound.

Building upon the derived Lipschitz constant for each individual weight matrix (e.g., $\mathbf{W}_{QK}, \mathbf{W}_{V}$, etc.),  we have established the sensitivity of the ECCT model to perturbations. Based on the fact that the output is Lipschitz continuous with respect to each weight, we demonstrate below that the entire decoder $f$ is also Lipschitz continuous with respect to the collective set of weights. This transition serves as a necessary step for estimating the overall covering number of the proposed model, as it allows us to characterize the stability of the entire parameter space.

\begin{theorem}[Lipschitzness of ECCT]\label{theorem: ECCT is Lipschitz continuous}
    Let the ECCT $f(\cdot;\mathcal{W}) \in \mathcal{F}_{ECCT}$ be a function with the weight set: $$\mathcal{W}=\{\mathbf{W}_{QK},\mathbf{W}_{V},\mathbf{W}_{F1},\mathbf{W}_{F2},\mathbf{W}_{o1},\mathbf{W}_{o2}\}.$$ For any two sets of weights $\mathcal{W}$ and $\mathcal{W}^{\prime}$, the variation in the decoder's $j$-th output is bounded as follows:
    \begin{equation}
        \left|f(\mathbf{X};\mathcal{W})[j]-f\left(\mathbf{X};\mathcal{W}^{\prime}\right)[j]\right| \leq \sum_{i \in \mathcal{I}} L_{i}\left\|\mathbf{W}_{i}-\mathbf{W}_{i}^{\prime}\right\|_{F},
        \label{eq: ECCT Lipschitz continuous}
    \end{equation}
    where $\mathcal{I}=\{Q K, V, F 1, F 2, o 1, o 2\}$ is the index set of the weights, and $L_{i}$ denotes the Lipschitz constant with respect to the weight matrix with index $i$, as derived in the preceding analysis.
\end{theorem}

\begin{proof}
    To characterize the impact of collective weight variations on the output, we employ a telescoping sum decomposition \cite{thomson2008elementaryrealanalysis}. We define a sequence of intermediate weight sets $\left\{\mathcal{W}^{(0)}, \mathcal{W}^{(1)}, \ldots, \mathcal{W}^{(6)}\right\}$ such that any two adjacent sets differ by exactly one weight matrix:
    \begin{itemize}
        \item $\mathcal{W}^{(0)}=\left\{\mathbf{W}_{Q K}^{\prime}, \mathbf{W}_{V}^{\prime}, \mathbf{W}_{F 1}^{\prime}, \mathbf{W}_{F 2}^{\prime}, \mathbf{W}_{o 1}^{\prime}, \mathbf{W}_{o 2}^{\prime}\right\}$ (corresponding to the set $\mathcal{W}^{\prime}$),
        \item $\mathcal{W}^{(1)}=\left\{\mathbf{W}_{Q K}, \mathbf{W}_{V}^{\prime}, \mathbf{W}_{F 1}^{\prime}, \mathbf{W}_{F 2}^{\prime}, \mathbf{W}_{o 1}^{\prime}, \mathbf{W}_{o 2}^{\prime}\right\}$,
        \item $\ldots$
        \item $\mathcal{W}^{(6)}=\left\{\mathbf{W}_{Q K}, \mathbf{W}_{V}, \mathbf{W}_{F 1}, \mathbf{W}_{F 2}, \mathbf{W}_{o 1}, \mathbf{W}_{o 2}\right\}$ (corresponding to the set $\mathcal{W}$).
    \end{itemize}
    The total difference between $f(\mathbf{X};\mathcal{W})[j]$ and $f\left(\mathbf{X};\mathcal{W}^{\prime}\right)[j]$ can be decomposed as the sum of incremental changes:
    \begin{equation}
        f(\mathbf{X};\mathcal{W})[j]-f\left(\mathbf{X};\mathcal{W}^{\prime}\right)[j] = \sum_{i=1}^{6}\left(f(\mathbf{X};\mathcal{W}^{(i)})[j]-f(\mathbf{X};\mathcal{W}^{(i-1)})[j]\right).
    \end{equation}
    By applying the triangle inequality, we obtain:
    \begin{equation}
        f(\mathbf{X};\mathcal{W})[j]-f\left(\mathbf{X};\mathcal{W}^{\prime}\right)[j] \leq \sum_{i=1}^{6}\left|f(\mathbf{X};\mathcal{W}^{(i)})[j]-f(\mathbf{X};\mathcal{W}^{(i-1)})[j]\right|.
    \end{equation}
    Based on the individual Lipschitz continuity established previously for each weight matrix, each absolute value term is bounded by the product of its corresponding Lipschitz constant and the Frobenius norm of the weight difference. Consequently, it follows that as \eqref{eq: ECCT Lipschitz continuous}. This completes the proof.
\end{proof}

\begin{proof}[Proof of Theorem~\ref{theorem: bound for single-layer ecct (main result)}]
It is enough to construct the matrix covering of each weight matrix in $\mathcal{W}$. Moreover, their Cartesian product yields a covering of the function class $\mathcal{F}_{ECCT}$ \cite{chen2020generalizationofafamilyofRNN,edelman2022inductivebiasTransformerGeneralization}. Therefore, the covering number of  $\mathcal{F}_{ECCT}$ can be upper bounded by the product of all the covering numbers of weight matrices\cite{bartlett2017spectrallynormalizedmarginbounds}:
\begin{equation}
        \mathcal{N}(\mathcal{F}_{ECCT}[j],\epsilon,\|\cdot\|_{2}) \leq \prod_{i\in \mathcal{I}} \mathcal{N}(\mathbf{W}_{i},\frac{\epsilon}{6L_{i}},\|\cdot\|_{F}), 
        \label{eq: global covering number upper bound (prod)}
\end{equation}
where $\mathcal{I}=\{Q K, V, F 1, F 2, o 1, o 2\}$ is the index set of the weights. We denote $\mathcal{N}(\mathcal{F}_{ECCT}[j],\epsilon,\|\cdot\|_{2})$ as the covering number of the function class $\mathcal{F}_{ECCT}$ at scale $\epsilon$ with respect to the $L_2$ norm at the $j$-th position. Similarly, for each weight matrix $\mathbf{W}_{i}$ within the parameter set $\mathcal{W}$, let $ \mathcal{N}(\mathbf{W}_{i},\frac{\epsilon}{6L_{i}},\|\cdot\|_{F})$ represent the covering number of the corresponding weight space at scale $\frac{\epsilon}{6L_{i}}$.

Using Lemma 8 in \cite{chen2020generalizationofafamilyofRNN}, we can write out the upper bounds of covering numbers of weight matrices as follows:

\begin{align}
    &\mathcal{N}\left(\mathbf{W}_{QK}, \frac{\epsilon}{6 L_{QK}}, \| \cdot \|_{F}\right) \leq\left(1+\frac{2 \sqrt{d} B_{QK} \cdot 6 L_{QK}}{\epsilon}\right)^{d^{2}},\label{eq: covering number of QK}\\
    &\mathcal{N}\left(\mathbf{W}_{V}, \frac{\epsilon}{6 L_{V}}, \| \cdot \|_{F}\right) \leq\left(1+\frac{2 \sqrt{d} B_{V} \cdot 6 L_{V}}{\epsilon}\right)^{d^{2}},\\
    &\mathcal{N}\left(\mathbf{W}_{F1}, \frac{\epsilon}{6 L_{F1}}, \| \cdot \|_{F}\right) \leq\left(1+\frac{2 \sqrt{d} B_{F1} \cdot 6 L_{F1}}{\epsilon}\right)^{ud^{2}}, \\
    &\mathcal{N}\left(\mathbf{W}_{F2}, \frac{\epsilon}{6 L_{F2}}, \| \cdot \|_{F}\right) \leq\left(1+\frac{2 \sqrt{d} B_{F2} \cdot 6 L_{F2}}{\epsilon}\right)^{ud^{2}},\\
    &\mathcal{N}\left(\mathbf{W}_{o1}, \frac{\epsilon}{6 L_{o1}}, \| \cdot \|_{F}\right) \leq\left(1+\frac{2 b_{o1} \cdot 6 L_{o1}}{\epsilon}\right)^{d}, \label{eq: covering number of o1}\\
    &\mathcal{N}\left(\mathbf{W}_{o2}[:,j], \frac{\epsilon}{6 L_{o2}}, \| \cdot \|_{F}\right) \leq\left(1+\frac{2 b_{o2} \cdot 6 L_{o2}}{\epsilon}\right)^{L}.\label{eq: covering number of o2}
\end{align}

Substituting \eqref{eq: covering number of QK}--\eqref{eq: covering number of o2} in \eqref{eq: global covering number upper bound (prod)}, we further obtain:

\begin{equation}
    \mathcal{N}(\mathcal{F}_{ECCT}[j],\epsilon,\|\cdot\|_{2}) \leq \left(1+\frac{12 \sqrt{d} \cdot B L^{2}}{\epsilon}\right)^{L+(2 u+2) d^{2}},
    \label{eq: covering number of global bound}
\end{equation}
where $B=b_{o1}^{2} L_{\sigma} B_{F1}^{2} B_{F2}^{2} B_{V}^{2} B_{QK} b_{x}^{3}w^{2}$.

For any $j$-th output of the ECCT decoder, we relate the covering number to the Rademacher complexity via Dudley’s entropy integral \cite{dudley1967sizesofcompactsubsetsofhilbertspace} for bounded-output functions. Specifically, for the function class $\mathcal{F}_{ECCT}$, by Lemma A.5 in \cite{bartlett2017spectrallynormalizedmarginbounds}, the Rademacher complexity of its $j$-th output can be upper-bounded as follows:
\begin{equation}
    R_{m}(\mathcal{F}_{ECCT}[j]) \leq \inf_{\alpha>0}\left(\frac{4 \alpha}{\sqrt{m}}+\frac{12}{m} \int_{\alpha}^{\sqrt{m}} \sqrt{\log \mathcal{N}\left(\mathcal{F}_{ECCT}[j], \epsilon,\|\cdot \|_{2} \right)}d \epsilon\right).
    \label{eq: Rademacher by Dudley inter bound}
\end{equation}

Substituting \eqref{eq: covering number of global bound} in \eqref{eq: Rademacher by Dudley inter bound}, the integral item can be bounded as follows,
\begin{equation}\begin{aligned}
    \int_{\alpha}^{\sqrt{m}} \sqrt{\log \mathcal{N}\left(\mathcal{F}_{ECCT}[j], \epsilon,\|\cdot \|_{2} \right)}d \epsilon &\leq \int_{\alpha}^{\sqrt{m}} \sqrt{(L+(2u+2)d^{2})\log\left(1+\frac{12 \sqrt{d} \cdot B L^{2}}{\epsilon}\right)}d \epsilon \\
    &\leq \int_{\alpha}^{\sqrt{m}} \sqrt{\left(L+(2u+2)d^{2}\right)\log\left(\frac{18 \sqrt{d} \cdot B L^{2}}{\epsilon}\right)}d \epsilon \\
    &\leq \sqrt{m(L+(2u+2)d^{2})\log\left(18 \sqrt{md} \cdot B L^{2}\right)},
    \label{eq: Dudley inter uppper bound}
\end{aligned}\end{equation}
where we assume that $\epsilon>0$ is small enough such that $\frac{6\sqrt{d} \cdot B L^{2}}{\epsilon}>1$. Then, by picking $\alpha=1/\sqrt{m}$, using \eqref{eq: Rademacher by Dudley inter bound}, we can obtain the upper bound as follows,
\begin{equation}
    R_{m}(\mathcal{F}_{ECCT}[j]) \leq \frac{4}{\sqrt{m}}+12 \sqrt{\frac{(L+(2 u+2) d^{2}) \log (18 \sqrt{m d}B L^{2})}{m}}.
    \label{eq: Rademacher bound for ECCT j}
\end{equation}

Finally, substituting \eqref{eq: Rademacher bound for ECCT j} in \eqref{eq: generalization bound via Rademacher(ECCT)}, we have
\begin{equation}
    \mathcal{R}_{\mathrm{BER}}(f)-\hat{\mathcal{R}}_{\mathrm{BER}}(f) \leq \frac{4}{\sqrt{m}}+\sqrt{\frac{\log (1 / \delta)}{2 m}}+12 \sqrt{\frac{(L+(2 u+2) d^{2}) \log (18 \sqrt{m d}B L^{2})}{m}}.
\end{equation}

\end{proof}

\section{Proof of Theorem \ref{theorem: impact of mask}} \label{appendix: proof of impact of mask (theorem2)}

In general, the derivation of the generalization bound for the ECCT decoder equipped with masked attention follows a trajectory similar to that of a standard ECCT decoder without masking, both rooted in the framework of covering numbers and Dudley’s entropy integral. However, a fundamental distinction arises: whether the attention mechanism, once constrained by a sparsity mask in the forward pass, strictly preserves the same sparsity structure in its gradients during backpropagation.

This property is the cornerstone of our theoretical analysis. The sparsity of the gradient is a necessary condition for refining the upper bound of the Lipschitz constant $L_{QK}$ with respect to the sparsity level $P$, rather than the full sequence length $L$. It is precisely this inheritance of sparsity from the forward pass to the gradient, as summarized in Lemma \ref{lemma: the sparsity of gradient df/dU}, that allows the theoretical advantages of the sparse mechanism to be manifested in the final generalization bound. In Appendix \ref{appendix: sparsity of masked attention gradient}, we provide a formal proof of this gradient sparsity.

We now proceed to compare the unmasked version (hereinafter referred to as the \emph{dense} version) with the masked version (referred to as the \emph{sparse} version). 
We begin with establishing a basic property of the Softmax gradient that holds for both dense and sparse settings. Let $S_{ij}$ denote the input to the row-wise Softmax function, such that $A_{i j}=\exp \left(S_{i j}\right) / \sum_{k} \exp \left(S_{i k}\right)$. For the dense version, $S_{ij}$ is identical to the logits $U_{ij}$, which is the entry of matrix $\mathbf{U}=\mathbf{X}\mathbf{W}_{QK}\mathbf{X}^{\top}\in  \mathbb{R}^{L\times L}$. The the gradient of $f$ with respect to $U_{ij}$ is:
\begin{equation}
    \begin{aligned}
    \frac{\partial f}{\partial U_{i j}}&=\sum_{k, l} \frac{\partial f}{\partial A_{k l}} \cdot \frac{\partial A_{k l}}{\partial U_{i j}}\\
    &=\sum_{l} \frac{\partial f}{\partial A_{i l}} \cdot \frac{\partial A_{il}}{\partial U_{i j}},
    \end{aligned}
\end{equation}
where the second equality follows from the fact that the row-wise Softmax operation depends only on the entries in the corresponding row of the matrix $\mathbf{U}$. Then, substituting the partial derivative of Softmax, $\frac{\partial A_{i l}}{\partial S_{i j}}=A_{i l}\left(\delta_{l j}-A_{i j}\right)$, we obtain:
\begin{equation}
    \begin{aligned}
\frac{\partial f}{\partial U_{i j}} & =\sum_{l=1}^{L} \frac{\partial f}{\partial A_{i l}}\left[A_{i l}\left(\delta_{l j}-A_{i j}\right)\right] \\
& =\sum_{l=1}^{L} \frac{\partial f}{\partial A_{i l}} A_{i l} \delta_{l j}-\sum_{l=1}^{L} \frac{\partial f}{\partial A_{i l}} A_{i l} A_{i j} \\
& =\frac{\partial f}{\partial A_{i j}} A_{i j}-A_{i j}\left(\sum_{l=1}^{L} \frac{\partial f}{\partial A_{i l}} A_{i l}\right),
\end{aligned}
\end{equation}
where $\delta_{lj}$ equals 1 if and only if $l=j$ otherwise 0.

Let $C_i = \sum_{l=1}^{L} \frac{\partial f}{\partial A_{i l}} A_{i l}$ denote the probability-weighted average of the gradients in row $i$, by the property of the row-wise Softmax operation. The gradient can be briefly written as:
\begin{equation}
    \frac{\partial f}{\partial U_{i j}} = A_{ij} \left( \frac{\partial f}{\partial A_{ij}} - C_i\right).
\end{equation}

Since $0< A_{ij} < 1$, we have $A_{ij}^2 \leq A_{ij}$. Thus:
\begin{equation}
    \sum_{l=1}^{L} \left(\frac{\partial f}{\partial U_{i j}}\right)^{2} =\sum_{l=1}^{L} A_{ij}^{2} \left( \frac{\partial f}{\partial A_{ij}} - C_i\right)^{2} \leq \sum_{l=1}^{L} A_{ij} \left( \frac{\partial f}{\partial A_{ij}} - C_i\right)^{2}.
\end{equation}

The right-hand side of the inequality is exactly the variance of the variable $\frac{\partial f}{\partial A_{ij}}(j=1,\ldots, L)$ under the probability distribution defined by $\mathbf{A}[i,:]$ (i.e., row $i$ of the matrix $\mathbf{A}$). Moreover, for a random variable $X$, we have $\operatorname{Var}(X)= \mathbb{E}(X^2)-[\mathbb{E}(X)]^2 \leq \mathbb{E}(X^2)$. Thus: 
\begin{equation}
    \sum_{j=1}^{L} A_{i j}\left(\frac{\partial f}{\partial A_{i j}}-C_{i}\right)^{2} \leq \sum_{j=1}^{L} A_{i j}\left(\frac{\partial f}{\partial A_{i j}}\right)^{2}.
\end{equation}

Again, using $A_{ij} < 1$:
\begin{equation}
    \sum_{j=1}^{L} A_{i j}\left(\frac{\partial f}{\partial A_{i j}}\right)^{2} \leq \sum_{j=1}^{L} \left(\frac{\partial f}{\partial A_{i j}}\right)^{2}.
\end{equation}
Therefore, for each row $i$, we have proven $\left\|\left(\nabla_{\mathbf{U}} f\right)[i,:]\right\|_{2}^{2} \leq\left\|\left(\nabla_{\mathbf{A}} f\right)[i,:]\right\|_{2}^{2}$.

Summing the squared norms over all rows $i=1, \ldots, L$, we can prove that:
\begin{equation}
    \left\|\nabla_{\mathbf{U}} f\right\|_{F}^{2}=\sum_{i=1}^{L}\left\|\left(\nabla_{\mathbf{U}} f\right)[i,:]\right\|_{2}^{2} \leq \sum_{i=1}^{L}\left\|\left(\nabla_{\mathbf{A}} f\right)[i,:]\right\|_{2}^{2}=\left\|\nabla_{\mathbf{A}} f\right\|_{F}^{2}.
    \label{eq: grad U f <= grad A f, Frobenius norm}
\end{equation}

This inequality serves as the common starting point for the subsequent analysis of both versions.

1) Lipschitz Constant $L_{QK}^{dense}$ for the Dense Version: For the unmasked (dense) decoder, we seek to bound $L_{QK}^{dense}\leq \sup \|\nabla_{\mathbf{W}_{QK}} f\|_F$ (by Lemma \ref{theorem: get Lipschitz constant}). Starting from $\mathbf{U}$, we apply the chain rule and the submultiplicativity of the Frobenius norm. Using \eqref{eq: grad U f <= grad A f, Frobenius norm}, we can expand the gradient bound as follows:
\begin{equation}
    \begin{aligned}
\left\|\nabla_{\mathbf{W}_{Q K}} f\right\|_{F} & \leq\|\mathbf{X}\|_{2}^{2}\left\|\nabla_{\mathbf{U}} f\right\|_{F} \\
& \leq\|\mathbf{X}\|_{2}^{2}\left\|\nabla_{\mathbf{A}} f\right\|_{F} \\
& \leq\|\mathbf{X}\|_{2}^{2}\left(\left\|\mathbf{X}\right\|_{F}\left\|\nabla_{\mathbf{H}_{att}} f\right\|_{F}\right),
\end{aligned}
\end{equation}
where $\mathbf{H}_{att}=\mathbf{X}^{\top}\mathbf{A}$ and its definition is consisent with that in Appendix \ref{appendix: bounding Lips constants}. Recursively bounding the upstream gradients ($\nabla_{\mathbf{H}_{att}}f$) and input terms, we obtain the result consistent with \eqref{eq: L_QK bound (dense)}:
\begin{equation}
    L_{QK}^{dense} \leq b_{o1}b_{o2} L_{\sigma} B_{F 2} B_{F 1} B_{V} L^{1.5} b_{x}^{3}.
\end{equation}

Under Assumption \ref{assumption: bounded input and weight}, we observe that among all the norm upper bounds, only $b_{o2}\leq w\sqrt{L}$ depends on the sequence length $L$. Substituting this bound yields:
\begin{equation}
    L_{QK}^{dense} \leq L^2 b_{x}^{3}wK,
    \label{eq: L_{QK} (dense, using w)}
\end{equation}
where $K=b_{o1} L_{\sigma} B_{F 2} B_{F 1} B_{V}$, and this bound is the same as \eqref{eq: L_QK bound (dense)}.

This refined bound is introduced for comparison with the subsequent sparse version. Specifically, our analysis shifts from the scale of the entire matrix to the element-wise scale. This refinement is necessary because the usage of sparsity requires a finer-grained upper bound to characterize the corresponding Lipschitz constant.

2) Lipschitz Constant $L_{QK}^{sparse}$ for the Sparse Version: 


To bound the element-wise partial derivative $|\frac{\partial f}{\partial A_{ij}}|$, we first examine the gradient flow from the ECCT decoder output back to the attention layer. Given the position-wise nature of the operations succeeding the attention mechanism (including FFN layers $F_1$,  $F_2$, and output projection $o1$), the $t$-th component before the final linear output layer ($\mathbf{H}_{o1}\in \mathbb{R}^{1\times L}$), depends only on the $t$-th column of $\mathbf{H}_{att}$ ($t=1,\ldots,L$). Staring from $\mathbf{H}_{o1}$, we have:
\begin{equation}
    \left|\left(\frac{\partial f}{\partial \mathbf{H}_{o1}}\right)_{t}\right|=\left| W_{o2, tj}\right|\leq w,
\end{equation}
where $W_{o2, tj}$ denotes the $t$-th element of $\mathbf{W}_{o2}[:,j]$, as we focus on the $j$-th final output bit. By applying the chain rule through the intermediate position-wise layers, the gradient with respect to the $t$-th column of $\mathbf{H}_{att}$ is determined as:
\begin{equation}
    \left\| (\nabla_{\mathbf{H}_{att}}f)_t \right\|_2 \leq wK.
\end{equation}

For any index $(i,j)$, the partial derivative of the output with respect to the attention score $A_{ij}$ is given by the inner product: $\frac{\partial f}{\partial A_{ij}}=\mathbf{x}_{i}^{\top}(\nabla_{\mathbf{H}_{att}}f)_j$. Therefore, we obtain the element-wise bound:
\begin{equation}
    \left|\frac{\partial f}{\partial A_{i j}}\right| \leq\left\|\mathbf{x}_{i}\right\|_{2}\left\|\left(\nabla_{H_{\text {att }}} g\right)_{j}\right\|_{2}\leq b_x wK,
\end{equation}
where $\mathbf{x}_{i} \in \mathbb{R}^{d \times 1}$ denotes the $i$-th column of the embedded input $\mathbf{X}$.

For the sparse verion, using \eqref{eq: grad U f <= grad A f, Frobenius norm} and Lemma \ref{lemma: the sparsity of gradient df/dU}, we have:
\begin{equation}
\begin{aligned}
    \left\|\nabla_{\mathbf{U}} f\right\|_{F} &\leq \sqrt{\sum_{(i,j)\in \Omega} \left(\frac{\partial f}{\partial A_{i j}}\right)^{2}} \\
    &\leq \sqrt{\sum_{(i,j)\in \Omega}\left(b_x w K\right)^2} \\
    &\leq \sqrt{LP}\cdot (b_x w K),
\end{aligned}
\end{equation}
where $\Omega$ denotes the set of indices preserved after masking.

Using this bound, we finally obtain:
\begin{equation}
\begin{aligned}
    L_{QK}^{sparse} &\leq \|X\|_{2}^{2}\left\|\nabla_{U} f\right\|_{F} \leq (\sqrt{L}b_x)^{2}\cdot(\sqrt{LP}\cdot (b_x w K)) \\
    &=L^{1.5}\sqrt{P}b_{x}^{3}wK.
\end{aligned}
\end{equation}

Upon obtaining the bound for $L_{QK}^{sparse}$, we observe that the Lipschitz bounds for the remaining constituent weight matrices remain consistent with the dense baseline, as detailed in Appendix \ref{appendix: bounding Lips constants}. By constructing the covering numbers for each individual weight space and considering their Cartesian product, we obtain the covering number for the entire decoder function class. This aggregation process yields the global Lipschitz bound $\Lambda^{sparse}$ defined in Theorem \ref{theorem: impact of mask}, which is strictly dominated by the contribution of $L_{QK}^{sparse}$ due to its higher-order dependency on the sequence length $L$.

Consequently, by following the same trajectory as the proof of Theorem \ref{theorem: bound for single-layer ecct (main result)} in Appendix \ref{appendix: proof of theorem 1}, the proof of Theorem \ref{theorem: impact of mask} is completed.

\section{Gradient Sparsity of Masked Attention} \label{appendix: sparsity of masked attention gradient}

Given an input $\mathbf{X}\in \mathbb{R}^{L\times d}$, let $\mathbf{U}=\mathbf{X}\mathbf{W}_{QK}\mathbf{X}^{\top}\in  \mathbb{R}^{L\times L}$ denote the attention logits before masking. We define $\Omega\subset \{1,\ldots,L\}\times\{1,\ldots,L\}$ as the set of indices preserved by the sparsity mask. For an entry $U_{i j}$ in $\mathbf{U}$, the result of the masking operation is defined as:
\begin{equation}
    S_{i j}=\left\{\begin{array}{ll}
U_{i j} & (i, j) \in \Omega \\
-\infty & (i, j) \notin \Omega
\end{array}\right. 
\end{equation}

The resulting sparse attention matrix $\mathbf{A}$ is obtained via a row-wise Softmax operation, which can be expressed as:
\begin{equation}
    A_{i j}=\operatorname{Softmax}(\mathbf{S})_{i j}=\frac{e^{S_{i j}}}{\sum_{k} e^{S_{i k}}},
\end{equation}
where $A_{ij}$ and $S_{ij}$ is the entry in $\mathbf{A}$ and $\mathbf{S}$, respectively.

\begin{lemma}\label{lemma: the sparsity of gradient df/dU}
    For the ECCT decoder function $f$, the gradient with respect to the pre-masking logits $\mathbf{U}$, denoted as $\nabla_{\mathbf{U}}f$, preserves the same sparsity pattern as $\Omega$.
\end{lemma}

\begin{proof}
$\nabla_{\mathbf{U}}f$ is computed via the chain rule:
\begin{equation}
    \begin{aligned}
    \frac{\partial f}{\partial U_{i j}}&=\sum_{k, l} \frac{\partial f}{\partial A_{k l}} \cdot \frac{\partial A_{k l}}{\partial U_{i j}}\\
    &=\sum_{k, l} \frac{\partial f}{\partial A_{k l}}\left(\sum_{m, n} \frac{\partial A_{k l}}{\partial S_{m n}} \cdot \frac{\partial S_{m n}}{\partial U_{i j}}\right).
    \end{aligned}
    \label{eq: df / dU_ij}
\end{equation}

The key point of our proof lies in analyzing the sparsity of the partial derivative $\frac{\partial A_{k l}}{\partial U_{i j}}$. Specifically, we aim to demonstrate that the Jacobian matrix inherits the sparsity structure of $\Omega$, ensuring that the gradient $\nabla_{\mathbf{U}}f$ remains $P$-sparse.

In \eqref{eq: df / dU_ij}, the term $\frac{\partial S_{m n}}{\partial U_{i j}}$ can be further written as :
\begin{equation}
    \frac{\partial S_{m n}}{\partial U_{i j}} = \delta_{mi} \delta_{nj}\mathbb{I}\left((i,j)\in \Omega \right),
    \label{eq: dSmn/dUij}
\end{equation}
where for any two non-negative integers $a$ and $b$, $\delta_{ab}$ denotes the Kronecker delta, which equals 1 if and only if $a=b$ and 0 otherwise. The indicator function $\mathbb{I}\left((i,j)\in \Omega \right)$ equals 1 if and only if the index $(i,j) \in \Omega$ and 0 otherwise. This implies that the partial derivative in \eqref{eq: dSmn/dUij} equals 1 if and only if $S_{m n}$ and $U_{i j}$ correspond to the same position (i.e., have the same index), which will not be masked, and equals 0 otherwise.

Therefore, $\frac{\partial A_{k l}}{\partial U_{i j}}$ in \eqref{eq: df / dU_ij} can be further simplified as follows:

\begin{equation}
    \frac{\partial A_{k l}}{\partial U_{i j}}=\frac{\partial A_{k l}}{\partial S_{i j}} \cdot \frac{\partial S_{i j}}{\partial U_{i j}}=\frac{\partial A_{k l}}{\partial S_{i j}} \cdot \mathbb{I}((i, j) \in \Omega).
\end{equation}

We now reduce the problem to analyzing $\frac{\partial A_{k l}}{\partial S_{i j}}$. Since $\mathbf{A}$ is obtained from $\mathbf{U}$ via the row-wise Softmax operation, the partial derivative $\frac{\partial A_{k l}}{\partial S_{i j}}$ can be expressed as:

\begin{equation}
    \frac{\partial A_{k l}}{\partial S_{i j}}=\delta_{k i} \cdot A_{k l}\left(\delta_{l j}-A_{i j}\right).
    \label{eq: dA_kl/S_ij}
\end{equation}

Clearly, when $k \neq i$, \eqref{eq: dA_kl/S_ij} equals 0. Therefore, it suffices to consider the case where $A_{k l}$ and $S_{i j}$ lie in the same row, which is consistent with the property of the row-wise Softmax function. Equation \eqref{eq: dA_kl/S_ij} can be further rewritten as:
\begin{equation}
    \frac{\partial A_{i l}}{\partial S_{i j}}= A_{i l}\left(\delta_{l j}-A_{i j}\right).
\end{equation}

Using the fact that $A_{i j}=0$ if $(i,j)\notin \Omega$, we can easily verify that $\frac{\partial A_{i l}}{\partial S_{i j}}=0$ if this position is masked. Therefore, \eqref{eq: df / dU_ij} can finally be written as:
\begin{equation}
    \begin{aligned}
        \frac{\partial f}{\partial U_{i j}} &= \sum_{l}\frac{\partial f}{\partial A_{il}} \cdot \frac{\partial A_{i l}}{\partial S_{i j}} \\
        &= \begin{cases}
            \sum_{l} \frac{\partial f}{\partial A_{il}}\cdot A_{i l}\left(\delta_{l j}-A_{i j}\right), & (i,j) \in \Omega \\
            0, & (i,j) \notin \Omega
        \end{cases}.
    \end{aligned}
\end{equation}

In summary, we have shown that the gradient $\nabla_{\mathbf{U}}f$ preserves the same sparsity pattern as $\Omega$. This completes the proof.

\end{proof}

\section{Bounding Lipschitz Constants}\label{appendix: bounding Lips constants}
We first present Theorem \ref{theorem: get Lipschitz constant}, a method for computing Lipschitz constants via the gradients of functions. We then exploit properties of matrix norms to derive upper bounds on the norms of intermediate variables, which in turn yield.

We then derive upper bounds on the norms of the relevant intermediate variables and, following the order of gradients in backpropagation together with the chain rule, sequentially obtain Lipschitz constants with respect to each weight matrix.

\begin{theorem}[See \cite{virmaux2018lipschitzConstantEstimation}]\label{theorem: get Lipschitz constant}
    If $f: \mathbb{R}^{n}\rightarrow \mathbb{R}^{m}$ is Lipschitz continuous, then its Lipschitz constant with respect to input $x$ is the maximum norm of its gradient on the domain set:
    \begin{equation}
        L(f) = \sup_{x} \|\nabla_x f \|_2.
    \end{equation}
\end{theorem}

\emph{Bounding Intermediate Variables}: 
\begin{itemize}
    \item $\mathbf{H}_{att}\triangleq \mathbf{X}^{\top} \operatorname{Softmax}\left(\mathbf{X} \mathbf{W}_{Q K}^{\top} \mathbf{X}^{\top}\right)=\mathbf{X}^{\top} \mathbf{A}$. By the properties of the row-wise softmax operation, we have $\|\mathbf{A}\|_2 \leq \sqrt{L}$, $\left\|\mathbf{H}_{att}\right\|_{F} \leq\left\|\mathbf{X}^{\top} \mathbf{A}\right\|_{F} \leq\left\|\mathbf{X}^{\top}\right\|_{F}\|\mathbf{A}\|_{2}  \leq (\sqrt{L} b_{x})\sqrt{L}=L b_x$. 
    \item $\mathbf{H}_{V}\triangleq \mathbf{W}_{V}^{\top} \mathbf{H}_{att}$, $\|\mathbf{H}_{V}\|_{F} \leq\left\|\mathbf{W}_{V}^{\top}\right\|_{2}\|\mathbf{H}_{att}\|_{F} \leq B_{V}\left(L b_{x}\right)$.
    \item $\mathbf{H}_{F1}\triangleq \mathbf{W}_{F1}^{\top} \mathbf{H}_{V}$, $\left\|\mathbf{H}_{F1}\right\|_{F} \leq\left\|\mathbf{W}_{F1}^{\top}\right\|_{2}\left\|\mathbf{H}_{V}\right\|_{F} \leq B_{F1}\left(B_{V} L b_{x}\right)$.
    \item $\mathbf{H}_{\sigma}\triangleq\sigma\left(\mathbf{H}_{F1}\right)$, $\left\|\mathbf{H}_{\sigma}\right\|_{F} \leq L_{\sigma}\left\|\mathbf{H}_{F1}\right\|_{F} \leq L_{\sigma} B_{F1} B_{V} L b_{x}$.
    \item $\mathbf{H}_{F2}\triangleq\mathbf{W}_{F2}^{\top} \mathbf{H}_{\sigma}$, $\left\|\mathbf{H}_{F2}\right\|_{F} \leq\left\|\mathbf{W}_{F2}^{\top}\right\|_{2}\left\|\mathbf{H}_{\sigma}\right\|_{F} \leq B_{F2}\left(L_{\sigma} B_{F_{1}} B_{V} L b_{x}\right)$.
    \item $\mathbf{H}_{o1}\triangleq\mathbf{W}_{o1}^{\top} \mathbf{H}_{F2}$, $\left\|\mathbf{H}_{o1}\right\|_{F} \leq\left\|\mathbf{W}_{o1}^{\top}\right\|_{2}\left\|\mathbf{H}_{F2}\right\|_{F} \leq b_{o1}\left(B_{F 2} L_{\sigma} B_{F 1} B_{V} L b_{x}\right)$.
    \item The $j$-th output: $f=\mathbf{H}_{o2}\triangleq \mathbf{H}_{o1} \mathbf{W}_{o2}[:,j]$.
\end{itemize}

\emph{Bounding Lipschitz Constants with Respect to Each Weight Matrix}: 

We now derive upper bounds on the Lipschitz constants, with each subscript consistently corresponding to its associated weight matrix.

\begin{itemize}
    \item $L_{o2}$: Since $\nabla_{\mathbf{W}_{o2}[:,j]}f=\mathbf{H}_{o1}^{\top}$, we have the upper bound as $L_{o2} \leq \|\mathbf{H}_{o1}^{\top}\|_2=\|\mathbf{H}_{o1}\|_F \leq b_{o1}B_{F 2} L_{\sigma} B_{F 1} B_{V} L b_{x}$.
    \item $L_{o1}$: Since $\nabla_{\mathbf{W}_{o1}}f=\mathbf{H}_{F2}(\nabla_{\mathbf{H}_{o1}}f)^{\top}$, we first need to bound $\nabla_{\mathbf{H}_{o1}}f$ as $\left\|\nabla_{\mathbf{H}_{o1}}f\right\|_{F}=\left\|\mathbf{W}_{o2}[:,j]^{\top}\right\|_{F}=\left\|\mathbf{W}_{o2}[:,j]\right\|_{2} \leq b_{o 2}$, and then we have $L_{o1} \leq \| \nabla_{\mathbf{W}_{o1}}f \|_2 \leq \|\mathbf{H}_{F2}\|_F \|(\nabla_{\mathbf{H}_{o1}}f)^{\top}\|_F \leq b_{o2}B_{F 2} L_{\sigma} B_{F 1} B_{V} L b_{x}$.
    \item $L_{F2}, L_{F1}, L_{V}$ and $L_{QK}$: Following a similar procedure, we obtain the following upper bounds on the Lipschitz constants with respect to the weight matrices:
    \begin{gather}
        L_{F2}  \leq b_{o1}b_{o2} L_{\sigma} B_{F 1} B_{V} L b_{x}, \\
        L_{F1} \leq b_{o1}b_{o2} L_{\sigma} B_{F 2} B_{V} L b_{x}, \\
        L_{V} \leq b_{o1}b_{o2} L_{\sigma} B_{F 2} B_{F 1} L b_{x}, \\
        L_{QK} \leq b_{o1} L_{\sigma} B_{F 2} B_{F 1} B_{V} L^2 b_{x}^{3}w. \label{eq: L_QK bound (dense)}
    \end{gather}
\end{itemize}

In the derivation above, we extensively exploit properties of the spectral norm and the Frobenius norm of matrix products. In particular, for matrix $\mathbf{A}$ and $\mathbf{B}$, we have $\|\mathbf{A}\mathbf{B}\|_F\leq \|\mathbf{A}\|_2 \|\mathbf{B}\|_F$ .

\section{Proof of Theorem \ref{theorem: multi-layer ecct bound}}

To derive the generalization bound for the multi-layer ECCT decoder $f^{(T)}$, it is insufficient to consider only the partial derivatives with respect to weights in $\mathcal{W}^{(i)}$ for the $i$-th attention layer. Due to the recursive structure of deep networks, for a weight $\mathbf{W}^{(i)} \in \mathcal{W}^{(i)}$, we have:
\begin{equation}
    \frac{\partial f}{\partial \mathbf{W}^{(i)}}=\frac{\partial f}{\partial \mathbf{X}^{(T)}} \cdot\left(\prod_{t=i+1}^{T} \frac{\partial \mathbf{X}^{(t)}}{\partial \mathbf{X}^{(t-1)}}\right) \cdot \frac{\partial \mathbf{X}^{(i)}}{\partial \mathbf{W}^{(i)}}.
\end{equation}

Let $\alpha^{(t)}$ denote the Lipschitz constant of the mapping $\phi^{(t)}: \mathbf{X}^{(t-1)} \rightarrow \mathbf{X}^{(t)}$. By Theorem \ref{theorem: get Lipschitz constant}, this constant bounds the propagation of gradients such that $\left\| \nabla_{\mathbf{X}^{(t-1)}} f \right\|_F \leq \alpha^{(t)} \left\| \nabla_{\mathbf{X}^{(t)}} f \right\|_F$.

To obtain $\alpha^{(t)}$, we analyze the specific structure of the sparse attention layer. Let the mapping $\phi^{(t)}$ be decomposed into the product of the attention branch $\mathbf{A}(\cdot)$ and the value branch $\mathbf{V}(\cdot)$, followed by the output projection. Formally, we express the output $\mathbf{X}^{(t)}$ as:
\begin{equation} 
\mathbf{X}^{(t)} = \mathbf{A}\left(\mathbf{X}^{(t-1)}\right) \cdot \mathbf{V}\left(\mathbf{X}^{(t-1)}\right) \cdot \mathbf{W}_{F 2}^{(t)}. 
\end{equation}

Since the input $\mathbf{X}^{(t-1)}$ influences both the attention $\mathbf{A}$ and and the value $\mathbf{V}$, we apply the product rule of differentiation to get $\alpha^{(t)}$:
\begin{equation}
    \alpha^{(t)} \leq \left\|\mathbf{W}_{F 2}^{(t)}\right\|_2\left(\left\|\mathbf{A}\right\|_2 \left\|\frac{\partial \mathbf{V}}{\partial \mathbf{X}^{(t-1)}}\right\|_2 + \left\|\frac{\partial \mathbf{A}}{\partial \mathbf{X}^{(t-1)}}\right\|_2 \left\|\mathbf{V}\right\|_2\right)L_{\sigma},
    \label{eq: alpha(t) upper bound by product of A V X}
\end{equation}
where we add the $L_{\sigma}$ term according to the Talagrand’s concentration lemma \cite{ledouxTalagrand2013probability}. First, by the properties of the row-wise Softmax function and the norms, we can easily obtain:
\begin{gather}
    \|\mathbf{A}\|_{2} \leq \sqrt{\|\mathbf{A}\|_{1}\|\mathbf{A}\|_{\infty}} \leq \sqrt{P \cdot 1}=\sqrt{P}, \label{eq: appendix proof 3, ||A||} \\
    \|\mathbf{V}\|_{2} \leq \|\mathbf{X}\|_{2}\left\|\mathbf{W}_{V}^{(t)}\right\|_{2}\left\|\mathbf{W}_{F 1}^{(t)}\right\|_{2} \leq b_{x} B_{V} B_{F 1}, \label{eq: appendix proof 3, ||V||}\\
    \left\|\frac{\partial \mathbf{V}}{\partial \mathbf{X}^{(t-1)}}\right\|_2 \leq \left\|\mathbf{W}_{V}^{(t)}\right\|_{2} \left\|\mathbf{W}_{F 1}^{(t)}\right\|_{2} \leq B_{V}B_{F1}. \label{eq: appendix proof 3, ||dV/dX||}
\end{gather}

Then, the derivative of the attention matrix $\mathbf{A}$ involves the bilinear form $\mathbf{X}^{(t-1)} \mathbf{W}_{Q K}^{(t)} \left(\mathbf{X}^{(t-1)}\right)^{\top}$, introducing a factor of 2. Furthermore, since the gradient support is restricted to the unmasked set $\Omega$ (see Appendix \ref{appendix: sparsity of masked attention gradient}), the Frobenius norm of the gradient scales with $\sqrt{P}$ rather than $\sqrt{L}$ relative to the row-wise bounds. Thus, we have
\begin{equation}
    \left\|\frac{\partial \mathbf{A}}{\partial \mathbf{X}^{(t-1)}}\right\|_2 \leq 2\sqrt{P} B_{QK}b_{x}. \label{eq: appendix proof 3, ||dA/dX||}
\end{equation}

Substituting \eqref{eq: appendix proof 3, ||A||}-\eqref{eq: appendix proof 3, ||dA/dX||} in \eqref{eq: alpha(t) upper bound by product of A V X}, we have:
\begin{equation}
    \alpha^{(t)} \leq \sqrt{P} \cdot L_{\sigma} B_{V} B_{F1}B_{F2} \left( 1 + 2 B_{QK} b_x^2 \right).
    \label{eq: alpha^(t), P-sparse}
\end{equation}

Again, using the chain rule, we obtain the Lipschitz constants with respect to weights in $\mathcal{W}^{(i)}$ for the $i$-th attention when $T>1$ and $i<T$:
\begin{gather}
    L_{Q K}^{(i)} \leq b_{o 1}L_{\sigma} B_{V} B_{F 1} B_{F 2}b_{x}^{3}w \cdot L^{1.5} \sqrt{P} \cdot\prod_{t=i+1}^{T}\alpha^{(t)}, \label{eq: L_QK, i-th layer} \\
     L_{F2}^{(i)}  \leq b_{o1}b_{o2} L_{\sigma} B_{F 1} B_{V} L b_{x} \cdot\prod_{t=i+1}^{T}\alpha^{(t)}, \\
     L_{F1}^{(i)} \leq b_{o1}b_{o2} L_{\sigma} B_{F 2} B_{V} L b_{x} \cdot\prod_{t=i+1}^{T}\alpha^{(t)}, \\
     L_{V}^{(i)} \leq b_{o1}b_{o2} L_{\sigma} B_{F 2} B_{F 1} L b_{x} \cdot\prod_{t=i+1}^{T}\alpha^{(t)}. \label{eq: L_V, i-th layer}
\end{gather}

For each weight within the set $\mathcal{W}^{(i)}(i=1,\ldots,T)$ and weights $\mathbf{W}_{o1}, \mathbf{W}_{o2}$, we can further obtain the upper bounds of covering numbers as follows:

\begin{align}
    &\mathcal{N}\left(\mathbf{W}_{QK}^{(i)}, \frac{\epsilon}{(4T+2) L_{QK}^{(i)}}, \| \cdot \|_{F}\right) \leq\left(1+\frac{2 \sqrt{d} B_{QK} \cdot (4T+2) L_{QK}^{(i)}}{\epsilon}\right)^{d^{2}},\label{eq: covering number of QK, i-th layer, multi-layer}\\
    &\mathcal{N}\left(\mathbf{W}_{V}^{(i)}, \frac{\epsilon}{(4T+2) L_{V}^{(i)}}, \| \cdot \|_{F}\right) \leq\left(1+\frac{2 \sqrt{d} B_{V} \cdot (4T+2) L_{V}^{(i)}}{\epsilon}\right)^{d^{2}},\\
    &\mathcal{N}\left(\mathbf{W}_{F1}^{(i)}, \frac{\epsilon}{(4T+2) L_{F1}^{(i)}}, \| \cdot \|_{F}\right) \leq\left(1+\frac{2 \sqrt{d} B_{F1} \cdot (4T+2) L_{F1}^{(i)}}{\epsilon}\right)^{ud^{2}}, \\
    &\mathcal{N}\left(\mathbf{W}_{F2}^{(i)}, \frac{\epsilon}{(4T+2) L_{F2}^{(i)}}, \| \cdot \|_{F}\right) \leq\left(1+\frac{2 \sqrt{d} B_{F2} \cdot (4T+2) L_{F2}^{(i)}}{\epsilon}\right)^{ud^{2}},\\
    &\mathcal{N}\left(\mathbf{W}_{o1}, \frac{\epsilon}{(4T+2) L_{o1}}, \| \cdot \|_{F}\right) \leq\left(1+\frac{2 b_{o1} \cdot (4T+2) L_{o1}}{\epsilon}\right)^{d}, \label{eq: covering number of o1, multi-layer}\\
    &\mathcal{N}\left(\mathbf{W}_{o2}[:,j], \frac{\epsilon}{(4T+2) L_{o2}}, \| \cdot \|_{F}\right) \leq\left(1+\frac{2 b_{o2} \cdot (4T+2) L_{o2}}{\epsilon}\right)^{L}.\label{eq: covering number of o2, multi-layer}
\end{align}

For the $j$-th output of $f^{(T)}$, the covering number can be bounded by the product of the covering numbers for each weight:
\begin{equation}
\begin{aligned}
        \mathcal{N}(\mathcal{F}_{ECCT,T}[j],\epsilon,\|\cdot\|_{2}) &\leq \mathcal{N}\left(\mathbf{W}_{o1}, \frac{\epsilon}{(4T+2) L_{o1}}, \| \cdot \|_{F}\right) \times \\
        &\mathcal{N}\left(\mathbf{W}_{o2}[:,j], \frac{\epsilon}{(4T+2) L_{o2}}, \| \cdot \|_{F}\right) \times \\ 
        &\prod_{i=1}^{T}\prod_{k\in \mathcal{I}_{Att}} \mathcal{N}(\mathbf{W}_{k}^{(i)},\frac{\epsilon}{(4T+2))L_{k}^{(i)}},\|\cdot\|_{F}), 
        \label{eq: global covering number upper bound (prod), multi-layer}
 \end{aligned}       
\end{equation}
where $\mathcal{I}_{Att}=\{QK, V,F1, F2\}$ denotes the index set of weights in an attention layer. We can further upper bound the covering number for the decoder as:
\begin{equation}
    \mathcal{N}(\mathcal{F}_{ECCT,T}[j],\epsilon,\|\cdot\|_{2})  \leq \left(1+\frac{4\sqrt{d}(2T+1)\Lambda^{(T)}}{\epsilon}\right)^{L+(2u+2)d^{2}T}
\end{equation}
where $\Lambda^{(T)}=b_{o1}L_{\sigma}B_{V}B_{F1}B_{F2}b_{x}^{3}wL^{1.5}\sqrt{P}\left(\sqrt{P}B_{V}B_{F1}B_{F2}(1+2B_{QK}b_{x}^{2})\right)^{T-1}$, which is assumed to be large enough to approximate the term inside the logarithm. Assuming $L \gg T$, and using \eqref{eq: Rademacher by Dudley inter bound} and \eqref{eq: generalization bound via Rademacher(ECCT)}, we can obtain the generalizaion bound:
\begin{equation}
    \mathcal{R}_{\mathrm{BER}}(f)-\hat{\mathcal{R}}_{\mathrm{BER}}(f) \leq \frac{4}{\sqrt{m}}+\sqrt{\frac{\log (1 / \delta)}{2 m}}+12 \sqrt{\frac{(L+(2 u+2) d^{2} T) \log (6 \sqrt{m d}\Lambda^{(T)}) }{m}}.
\end{equation}

Based on our previous derivations, the global Lipschitz bound $\Lambda^{(T)}$ for the sparse version of $T$-layer ECCT is dominated by the Lipschitz constant of the bottom-most attention layer ($i=1$). Specifically, under the $P$-sparse constraint, we have established that $\alpha^{(t)} \propto \sqrt{P}$ and $L_{QK}^{(T)}\propto L^{1.5}\sqrt{P}$. Aggregating these through $T-1$ layers, the bound of the sparse version scales as $\Lambda^{(T)} = \mathcal{O} \left( L^{1.5}\sqrt{P} \cdot (\sqrt{P})^{T-1} \right) = \mathcal{O} \left( L^{1.5} P^{T/2} \right)$.

For the dense version, each query attends to all $L$ keys in a row, meaning the spectral norm of the row-wise softmax matrix relaxes from $\left\|\mathbf{A}\right\|_2\leq \sqrt{P}$ to $\left\|\mathbf{A}\right\|_2\leq \sqrt{L}$. Similarly, the gradient support of the attention scores expands from $LP$ to $L^2$ entries, causing the functional Lipschitzness to scale with $L^2$ (i.e, $ L^{1.5}\sqrt{L}$). Consequently, the global bound for the dense version is recovered by replacing all instances of $\sqrt{P}$ with $\sqrt{L}$ in our framework: $\Lambda^{dense,(T)} = \mathcal{O} \left( L^{1.5}\sqrt{L} \cdot (\sqrt{L})^{T-1} \right) = \mathcal{O} \left( L^{1.5} L^{T/2} \right)$.

By comparing these two bounds, we define the contraction factor $\eta^{(T)}$ to characterize the reduction in the hypothesis space volume:
\begin{equation}
    \eta^{(T)} \triangleq \frac{\Lambda^{(T)}}{\Lambda^{dense,(T)}} \leq \left( \sqrt{\frac{P}{L}} \right)^T.
\end{equation}

This completes the proof of Theorem \ref{theorem: multi-layer ecct bound}.

\section{Proof of Theorem \ref{theorem: bound for AWGN}}

For brievty, we define the generalization gap as $\Delta(f)=\mathcal{R}_{\text{BER}}(f)-\hat{\mathcal{R}}_{\text{BER}}(f)$. Using the law of total expectation, we decompose the expected error based on the occurrence of event $\mathcal{E}$:
\begin{equation}
    \mathbb{E}\left[\Delta\right]=\mathbb{E}\left[\Delta\mid \mathcal{E}\right] \mathrm{Pr}\left(\mathcal{E}\right) + \mathbb{E}\left[\Delta\mid \mathcal{E}^c\right] \mathrm{Pr}\left(\mathcal{E}^c\right),
    \label{eq: total expectation of unbounded input}
\end{equation}
where $\mathcal{E}=\left\{\forall i \in[L],\left\|\mathbf{X}[i,:]\right\|_{2} \leq b_{x}\right\}$ and $\mathcal{E}^c$ denots its complement. 

According to the definition of $\mathbb{E}\left[\Delta\mid \mathcal{E}^c\right] \mathrm{Pr}\left(\mathcal{E}^c\right)$, we have:
\begin{equation}
    \begin{aligned}
    \mathbb{E}\left[\Delta \mid \mathcal{E}^{c}\right] \operatorname{Pr}\left(\mathcal{E}^{c}\right) & \leq \sup _{\mathcal{S} \in \mathcal{E}^{c}}|\mathcal{R}_{\text{BER}}(f)-\hat{\mathcal{R}}_{\text{BER}}(f)| \cdot \operatorname{Pr}\left(\mathcal{E}^{c}\right) \\
    & \leq M_{\max } \cdot \operatorname{Pr}\left(\mathcal{E}^{c}\right)
\end{aligned}
\end{equation}
Since the BER loss is strictly bounded by $1$, the generalization gap on the unbounded domain is trivially bounded by the worst-case constant $ M_{\max } =1$. Thus, $\left|\mathbb{E}\left[\Delta \mid \mathcal{E}^{c}\right]\right|\leq 1$. This simplifies the second term of the total expectation in \eqref{eq: total expectation of unbounded input}.

Then, we obtain the following inequality using the fact that any probability is bounded by 1:
\begin{equation}
     \mathbb{E}\left[\Delta\right]\leq\mathbb{E}\left[\Delta\mid \mathcal{E}\right] +  \mathrm{Pr}\left(\mathcal{E}^c\right).
     \label{eq: total expectation of unbounded input, v2}
\end{equation}

For the $i$-th row of input $\mathbf{X}$ (i.e., $\mathbf{X}[i,:]$), we have $\left\|\mathbf{X}[i,:]\right\|_2=|\tilde{y_i}|\cdot \left\|\mathbf{W}_{emb}[i,:]\right\|_2$, where $\tilde{y_i}$ denotes the $i$-th element of the $\tilde{\mathbf{y}}$, and $\mathbf{W}_{emb}[i,:]$ denotes the $i$-th row of the $\mathbf{W}_{emb} \  (i=1,\ldots, L)$. Since only the first $n$ rows of $\mathbf{X}$ are derived from the magnitude of $\mathbf{y}=(y_1,\ldots,y_n)$, which is affected by the noise level, we restrict our analysis to these rows. Thus, for $i=1,\ldots, n$, we have:
\begin{equation}
    \begin{aligned}
    \mathrm{Pr}\left(\left\|\mathbf{X}[i,:]\right\|_{2}>b_{x}\right)&=\mathrm{Pr}\left(\left|y_{i}\right|>\frac{b_{x}}{B_{e m b}}\right) \\
    &= Q(\frac{\tau-x^s_i}{\rho}) + Q(\frac{\tau+x^s_i}{\rho}) \ (\text{Let } \tau=\frac{b_{x}}{B_{e m b}}) \\
    &=Q(\frac{\tau-1}{\rho}) + Q(\frac{\tau+1}{\rho}) (\text{BPSK}),
    \end{aligned}
\end{equation}
where $\rho$ denotes the standard deviation of the AWGN. For the last $r$ rows of $\mathbf{X}$, they are always bounded since the syndromes are strictly bounded by $1$.

Using the union bound, we have:
\begin{equation}
    \mathrm{Pr}\left(\mathcal{E}^c\right) \leq n\left[ Q(\frac{\tau-1}{\rho}) + Q(\frac{\tau+1}{\rho}) \right].
    \label{eq: probability of E^c}
\end{equation}

Substituting \eqref{eq: probability of E^c} in \eqref{eq: total expectation of unbounded input, v2}, and applying the upper bound of the bounded input $\mathbb{E}\left[\Delta\mid \mathcal{E}\right]$ of Theorem \ref{theorem: multi-layer ecct bound} completes the proof of Theorem \ref{theorem: bound for AWGN}.

\bibliographystyle{IEEEtran}
\bibliography{refs}












\end{document}